%% file: main.tex
\documentclass{article}

\PassOptionsToPackage{numbers}{natbib}
\usepackage[preprint]{neurips_2025}

\usepackage[utf8]{inputenc} % allow utf-8 input
\usepackage[T1]{fontenc}    % use 8-bit T1 fonts
\usepackage{hyperref}       % hyperlinks
\usepackage{url}            % simple URL typesetting
\usepackage{booktabs}       % professional-quality tables
\usepackage{amsfonts}       % blackboard math symbols
\usepackage{nicefrac}       % compact symbols for 1/2, etc.
\usepackage{microtype}      % microtypography
\usepackage{xcolor}         % colors
\usepackage{float}

\usepackage{graphicx}
\usepackage{subcaption}

\usepackage{algorithm}
\usepackage[noend]{algpseudocode}

\usepackage{amsmath, amssymb, amsthm, mathtools}
\usepackage{bm}
\usepackage{float}% http://ctan.org/pkg/float

\DeclareMathOperator*{\argmin}{arg\,min}

\newcommand{\Var}{\operatorname{Var}}
\newcommand{\Cov}{\operatorname{Cov}}

\DeclareRobustCommand{\IPI}{\textnormal{IPI}}
\DeclareRobustCommand{\CIPI}{\textnormal{CIPI}}

\newtheorem{theorem}{Theorem}
\newtheorem{proposition}{Proposition}
\newtheorem{assumption}{Assumption}
\newtheorem{lemma}{Lemma}

\newtheorem{corollary}{Corollary}
\newtheorem{definition}{Definition}

\newtheorem{fact}{Fact}

\definecolor{red2}{rgb}{0.7, 0, 0.1}

\title{Imputation-Powered Inference}

\author{%
  Sarah Zhao\\
  Department of Statistics\\
  Stanford University\\
  Stanford, CA 94305 \\
  \texttt{smxzhao@stanford.edu} \\
  \And
  Emmanuel Candès \\
  Department of Statistics \\
  Stanford University\\
  Stanford, CA 94305 \\
  \texttt{candes@stanford.edu} \\
}

\begin{document}

\maketitle

\begin{abstract}
Modern multi-modal and multi-site data frequently suffer from blockwise missingness, where subsets of features are missing for groups of individuals, creating complex patterns that challenge standard inference methods. Existing approaches have critical limitations: complete-case analysis discards informative data and is potentially biased; doubly robust estimators for non-monotone missingness—where the missingness patterns are not nested subsets of one another—can be theoretically efficient but lack closed-form solutions and often fail to scale; and blackbox imputation can leverage partially observed data to improve efficiency but provides no inferential guarantees when misspecified. To address the limitations of these existing methods, we propose imputation-powered inference (IPI), a model-lean framework that combines the flexibility of blackbox imputation with bias correction using fully observed data, drawing on ideas from prediction-powered inference and semiparametric inference. IPI enables valid and efficient M-estimation under missing completely at random (MCAR) blockwise missingness and improves subpopulation inference under a weaker assumption we formalize as first-moment MCAR, for which we also provide practical diagnostics. Simulation studies and a clinical application demonstrate that IPI may substantially improve subpopulation efficiency relative to complete-case analysis, while maintaining statistical validity in settings where both doubly robust estimators and naive imputation fail to achieve nominal coverage.
\end{abstract}

\section{Introduction}
In modern datasets, missing values are pervasive, arising from multi-modal or multi-source data integration and heterogeneous participation across substudies \cite{sudlow2015ukbiobank, all2019allofus, blockwise_ad_multisource, multisource_neuroimaging}. We focus on a common and practically important structure: blockwise missingness, where entire subsets of variables are jointly missing in recurring patterns across samples. This type of missingness can severely undermine both the statistical efficiency and validity of downstream inference tasks that rely on complete data – even simple regression requires full observation of all outcomes and pertinent covariates.

In many settings, collecting the missing data is costly, time-consuming, or infeasible due to ethical or logistical constraints \cite{little2019statisticalmissing, census_survey_folktables}. A common workaround is to use machine learning models to impute missing entries. However, treating these imputed values as if they were observed
can lead to invalid inference and distorted conclusions.

To address these challenges, we introduce \textit{imputation-powered inference} (IPI), a model-lean framework that leverages flexible blackbox imputation models while ensuring valid statistical inference for a parameter of interest. Our approach integrates tools from semiparametric theory—particularly doubly robust estimation for non-monotone missingness \cite{robins1994, tsiatis2006semiparametric, sun2018inverse, sun2018inverse2}—with scalable and flexible blackbox models. This strategy builds on recent developments in prediction-powered inference (PPI) \cite{ppi}, which have primarily focused on inference in semi-supervised or single-pattern missingness where the missingness is either independent of the underlying data (\textit{missing completely at random}, or MCAR) or explicitly known and estimable by experimental design (\textit{missing at random}, or MAR, with known propensity scores) \cite{ppi, ji2025predictionssurrogatesrevisitingsurrogate_reppi, kluger2025predictionpoweredinferenceimputedcovariates}. In contrast, IPI extends blackbox-powered inference to more general blockwise missingness structures, where missing data may affect both covariates and outcomes in complex, structured patterns. By leveraging high-dimensional predictors—whether fully or partially observed—IPI improves imputation accuracy, yielding greater statistical efficiency while enabling valid inference for low-dimensional M-estimators associated with smooth convex losses, such as means or generalized linear model regression coefficients.

Under MCAR, IPI provides valid population-level inference and improves efficiency over complete-case analysis. When MCAR does not hold, our inferential goal shifts to the subpopulation for whom all variables of interest are observed. This subpopulation is particularly relevant in settings focused on high-risk or clinically prioritized groups—such as patients selected for additional testing—where researchers may wish to augment analyses using data from other sites or historical cohorts. However, combining such datasets often leads to complex, non-monotone missingness patterns. In these contexts, we demonstrate that IPI can still effectively utilize partially observed cases and remains valid under a weaker assumption, which we term first-moment MCAR. We formalize this condition and propose a straightforward diagnostic tool to assess its plausibility in practice.

Our contributions are threefold:
\begin{enumerate}
\item We introduce the \textit{imputation-powered inference} (IPI) and \textit{cross-imputation-powered inference} (CIPI) frameworks for scalable, valid inference for low-dimensional M-estimators under MCAR (Section \ref{sec:mcar_methodology}).
\item We analyze the asymptotic behavior of our estimators under non-MCAR blockwise missingness settings and introduce both the concept of and a diagnostic test for the first-moment MCAR assumption (Section \ref{sec:moving_away_mcar}).
\item We demonstrate IPI's substantial gains in inference efficiency on both simulated and real-world clinical datasets (Section \ref{sec: expts}).
\end{enumerate}

Code is provided at \href{https://github.com/sarmxzh/imputation-powered-inference}{https://github.com/sarmxzh/imputation-powered-inference}.

\section{Problem set-up}
\paragraph{Data and missingness patterns.}
Consider data $\{(D_i,M_i)\}_{i=1}^{N}$ drawn independently and identically distributed (i.i.d.) from a joint distribution $P_{D,M}$, where   
$D_i\in\mathbb{R}^{d}$ is the complete set of features and the missingness indicator vector  
$M_i\in\{1,\texttt{NA}\}^{d}$ specifies which entries are observed.

The \emph{available} data are  
\begin{equation}
O_i \;=\;D_i\circ M_i, 
\qquad\text{with}\quad
(O_i)_j \;=\;
\begin{cases}
D_{ij}, & M_{ij}=1,\\[4pt]
\texttt{NA}, & M_{ij}=\texttt{NA},
\end{cases}\quad j\in[d],
\end{equation}
where ``$\circ$’’ denotes the Hadamard (entry-wise) product and $[d] = \{1, ..., d\}$.

We consider \textit{blockwise} missingness with $R$ distinct nontrivial patterns $m_{1},\dots,m_{R}\in\{1,\texttt{NA}\}^{d}$, where $R$ is typically small relative to the sample size $N$, as in multi-site or multi-modal data fusion (see Figures \ref{fig:fm-results}(a) and \ref{fig:acs-results}(a) for examples of this blockwise missingness structure). The missingness distribution can then be modeled as a categorical variable over $R+1$ pattern indicators:
\begin{align*}
    M_i\sim\text{Categorical}\bigl((m_0 = \bm 1_d,m_{1},\dots,m_{R}),\,(p_{0},p_{1},\dots,p_{R})\bigr),
\qquad p_r := \mathbb{P}[M = m_r]>0.
\end{align*}
The joint distribution is then $P_{D,M} = P_M \times P_{D|M} = \prod\nolimits_{r =0}^R \left(p_r\times P_{D|m_r} \right)^{\bm{1}\{M = m_r\}}$, and under MCAR, $D$ is independent of $M$, so $P_{D|M} = P_D$.

We denote the number of samples for each missing pattern as 
\begin{equation}
    n :=\;\#\{i:M_i=\bm 1_d\},\qquad  
\tilde N_r :=\;\#\{i:M_i=m_r\},\qquad \tilde{N} := \;\sum\nolimits_{r=1}^{R}\tilde N_r,
\end{equation}
with $N$ denoting the total number of samples: $N=n + \tilde{N}$. For any function $g$ taking in observed data with inputs in $\{\mathbb{R} \cup \texttt{NA}\}^d$, we denote the empirical means for each missing pattern as
\begin{equation}
    \mathbb{P}_n g(O) = \frac{1}{n}\sum\nolimits_{i \in [N]: M_{i } = \bm{1}_d} g(O_i),\qquad  
    \mathbb{P}_{\tilde{N}_r}g(O) = \frac{1}{\tilde{N}_r} \sum\nolimits_{i \in [N]: M_i = m_r} g(O_i).
\end{equation}
We assume that $n, \tilde{N_r} > 0$ for $r \in [R]$ are sufficiently large for normal approximations to hold.

\paragraph{Covariates and auxiliary features.}
While the full feature set $D$ may be high-dimensional, our inferential target involves a lower-dimensional subset $X$. For clarity, we partition $D=(X,U)$, where $X\in\mathbb{R}^p$ contains the features of primary interest (this critically  \textit{includes} the response in regression and classification settings) and $U\in\mathbb{R}^{d-p}$ (possibly empty) consists of auxiliary features used solely to boost imputation. Missingness may occur in either $X$ or $U$, but subsetting $U$ to have less missingness can sometimes enhance imputation accuracy and overall sample efficiency.

\paragraph{Object of inference.}
Let $X$ denote the fully observed vector for features of interest and $\ell(X;\theta)$ be a specified convex loss. This work focuses on inference of the M-estimand for the fully observed population:
\begin{equation}
    \tilde{\theta}_\star = \argmin\nolimits_\theta \mathbb{E}[\ell(X; \theta) | M = \mathbf{1}_d],
\end{equation}
where the expectation is taken with respect to the distribution of $X$ conditional on there being no missingness in the feature vector. Estimands include subpopulation means, medians, quantiles and so on. Losses also include generalized linear model likelihoods, which include common estimands such as linear and logistic regression coefficients.

Under the MCAR setting (Section \ref{sec:mcar_methodology}), since $X$ is independent of $M$, $\tilde{\theta}_\star = \theta_\star$, the general population M-estimand:
\begin{equation}
    \theta_\star = \argmin\nolimits_\theta \mathbb{E}[\ell(X; \theta)],
\end{equation}
where the expectation is taken with respect to the true distribution of $X$. 

When we drop the MCAR assumption in Section \ref{sec:moving_away_mcar}, the shifted estimand $\tilde{\theta}_\star$ will often not be equal to $\theta_\star$. Here, $\tilde{\theta}_\star$ remains the primary object of inference when interest lies specifically in flagged subpopulations with more thorough measurements—such as a rare-disease focus group or patients who undergo further testing—rather than the full cohort.

\paragraph{Blackbox imputation.}
We assume access to a row-wise imputation method  
\begin{equation*}
    f:\bigl(\mathbb{R}\cup\{\texttt{NA}\}\bigr)^{d}\longrightarrow\mathbb{R}^{p},
\end{equation*}
which takes in the $d$-dimensional observed data vector $O$ (may contain \texttt{NA} entries) and fills in the missing entries for any missing entries in $X$. We remark that global procedures that exploit the \emph{entire} dataset at once must be adapted to variants that fit row-wise models on a training subset—such as MissForest \cite{stekhoven2012missforest}, EM methods \cite{dempster1977maximum}, or optimal-transport imputers \cite{muzellec2020missing_OTimpute}. 

\section{Related Work}\label{sec:related-works}

\paragraph{Prediction-powered inference.}
This work draws much of its inspiration from the recent line of PPI literature \cite{ppi, zrnic_cppi, angelopoulos2024ppi++, zrnic2024activestatisticalinference, kluger2025predictionpoweredinferenceimputedcovariates, zrnic2024notepredictionpoweredbootstrap, wang2024constructing_causalppi}, where the main idea is to leverage a fully observed dataset to debias flexible and powerful blackbox predictions in the context of semi-supervised settings. Recent variants exploit connections to surrogate outcomes \cite{ji2025predictionssurrogatesrevisitingsurrogate_reppi} or bootstrap corrections \citep{kluger2025predictionpoweredinferenceimputedcovariates} to provide more efficient or flexible methods. Crucially, all of these methods focus on a \textit{single} missing pattern (typically the response variable). In contrast, we allow for general blockwise missing patterns present in common multi-site and multi-modal data fusion applications.  

\textbf{Semiparametric theory and doubly robust estimators.}
Another line of work from which this paper draws inspiration is semiparametric theory and augmented inverse‐probability weighted (AIPW) estimators for efficient estimation under \textit{missingness-at-random} (MAR) \cite{bickel1993efficient,chernozhukov2024doubledebiasedmachinelearningtreatment,klaassen1987consistent,hasminskii1979nonparametric, levit1976efficiency,newey1994asymptotic, robinson1988root, robins1994, chen2000}. Our approach can be thought of as an instance of these estimators under the MCAR setting; however, our approach deviates in two main ways. Firstly, we provide a simple plug-and-play approach to boost inference compared to the complete-case estimator. This approach differs from the common semiparametric inference approaches discussed \cite{robins1994, tsiatis2006semiparametric}, which are often computationally intractable for today's complex high-dimensional data. Section \ref{expt:mcar} and Appendix \ref{app:aipw_expts} explore this tradeoff in an MCAR factor model simulation setting. Secondly, rather than aiming to protect against violations of MCAR through modeling MAR propensity scores and double machine learning assumptions, we ask what remains identifiable by our estimator when MCAR does not hold. This line of inquiry leads us to provide a minimal and testable assumption that helps characterize the limits of inference in more general missingness settings, without requiring correct specification of the imputation or missingness model.

\paragraph{Single and multiple imputation with ML.}
In practice, analysts frequently replace missing values with predictions from flexible algorithms—random‐forest imputation \citep{stekhoven2012missforest}, chained equations \citep{van2011mice}, OT-based methods \cite{muzellec2020missing_OTimpute}, or deep generative models \citep{yoon2018gain, mattei2019miwae, ivanov2018variational}. These tools sometimes excel at prediction accuracy, yet treating the filled‐in data as truly observed values can introduce severe bias in downstream inference. These impute-then-infer methods are referred to as single-imputation methods and often do not come with valid inference guarantees.

Foundational work by \citet{rubin1987multiple} formalized \emph{multiple imputation} (MI), demonstrating how repeated draws from a generative model can propagate uncertainty in downstream estimates.  MI remains the gold standard when (i) the analyst can posit a correctly specified joint model and (ii) the untestable MAR assumption holds. This approach; however, has been demonstrated to be fragile to violations of these assumptions \cite{bartlett2020bootstrapmi, murray2018multipleimputation, mccaw2024synsurr}. Our work instead aims to loosen these two assumptions about the correct specification of the model and about MAR -- the tradeoff is shifting the object of inference from the full population to the subpopulation that is fully observed in the non-MCAR case.

\textbf{Blockwise missing data.} We study the blockwise missing-data setting, which often arises when integrating multiple modalities or sources. Most prior methods for this regime target (generalized) linear models \citep{song2024semi,xue2021statistical,jin2023modular} or rely on specific missingness structures \citep{yu2020optimal,li2024adaptiveefficientlearningblockwise}. In contrast, IPI makes no assumptions about which covariates or outcomes are missing and accommodates general M-estimation with smooth convex losses.

\section{Imputation-powered inference (IPI) under MCAR}\label{sec:mcar_methodology}
To clarify the motivation and theoretical underpinnings of our framework, we introduce imputation-powered inference (IPI) and cross-imputation-powered inference (CIPI) under the MCAR assumption, presenting them as natural generalizations of PPI to blockwise missingness settings. Under MCAR, the IPI estimator can also be interpreted as a special case of AIPW estimators for non-monotone missingness in semiparametric inference. In contrast to prior work in semiparametrics, which emphasizes constructing semiparametrically efficient estimators \cite{robins1994, tsiatis2006semiparametric, sun2018inverse, sun2018inverse2}, our focus is on developing an implementation that is simple, scalable, and more stable. An empirical comparison with existing AIPW estimators is provided in Section~\ref{expt:mcar} and Appendix \ref{app:aipw_expts}, and a detailed discussion of their similarities and differences can be found in Appendix~\ref{app:semiparam}.

\subsection{Motivation for IPI from the PPI viewpoint}
We aim to leverage the main workflow of PPI: (i) enlarge the dataset with blackbox predictions, and (ii) debias it using fully labeled samples to ensure valid inference. Our aim is to increase effective sample size—the number of fully labeled observations that would yield the same precision of our method—while retaining nominal coverage. Indeed, intuitively, if the ML model can faithfully predict the missing values, we essentially have $\tilde{N}$ instead of $n$ datapoints, drastically boosting sample efficiency and inference power when $n \ll \tilde{N}$. If the model is misspecified, the debiasing step ensures unbiasedness and maintains coverage; we still harvest efficiency to the extent the predictions carry signal, without introducing bias.

Multiple missing-pattern settings leave (i) unchanged but complicate (ii) because imputation accuracy varies by pattern. A simple fix is to recalibrate each imputed sample against the complete cases by estimating the bias–we apply the sample’s missing pattern to fully observed data and compare the imputed values to the true ones:

\begin{align*}
    L_{\text{IPI pooled}} := \frac{1}{\tilde{N}} \sum_{i = 1}^{\tilde{N}} \bigg(\underbrace{\ell(f(O_i); \theta)}_{\text{imputed loss for partially observed sample $i$}} + \underbrace{\mathbb{P}_n \left[\ell(X; \theta) - \ell(f(O \circ M_i); \theta)\right]}_{\text{debiasing each imputation with fully observed dataset}}\bigg).
\end{align*}

Grouping by missing patterns gives
\begin{align*}
    L_{\text{IPI pooled}} = \sum_{r =1}^R \left(\frac{\tilde{N}_r}{\tilde{N}}\mathbb{P}_{\tilde{N}_r} \ell(f(O); \theta) + \mathbb{P}_n \left[\frac{1}{R} \ell(X; \theta) - \frac{\tilde{N}_r}{\tilde{N}}\ell(f(O \circ m_r); \theta)\right] \right).
\end{align*}

Setting $\tilde{\lambda}_r := R {\tilde{N}_r}/{\tilde{N}}$, we have a form akin to PPI++ \cite{angelopoulos2024ppi++}. This pooled loss can be efficient when common missing patterns are well imputed; however, that is not always the case—motivating our proposed estimator in the next section.

\subsection{IPI under MCAR}
We propose the following IPI loss with general tuning parameter $\lambda$.
\begin{definition} [IPI loss] Given $\theta \in \mathbb{R}^p$ and ${\lambda} = [\lambda_1, ..., \lambda_R]^\top \in \mathbb{R}^R$, we define the associated loss as:
\begin{equation}
    L_{\text{\em IPI}}(\theta; \lambda) 
    := \frac{1}{R}\sum_{r =1}^R \bigg(\lambda_r \mathbb{P}_{\tilde{N}_r} \left[\ell(f(O); \theta)\right]  + \mathbb{P}_{n} \bigg[\ell(X; \theta) - \lambda_r \ell(f(O \circ m_r); \theta)\bigg] \bigg).
\end{equation}
\end{definition}
One might first try to estimate $\theta_\star$ by minimizing $L_{\text{IPI}}$. However, unlike the PPI \cite{ppi} and PPI++ \cite{angelopoulos2024ppi++} semi-supervised objectives, arbitrary missingness in $D$ leads to a generally non-convex empirical risk—even for GLMs.

Under MCAR, though, $L_{\text{IPI}}$ is an unbiased estimator of $\mathbb{E}[\ell(X;\theta)]$, so a one-step update \cite{Vaart_1998, cam1960locally} recovers consistency and asymptotic normality, which we use for valid inference.

Concretely, we begin with a $\sqrt{n}$-consistent estimate—e.g., the complete-case estimator $\hat{\theta}_n$ in Eq.~\eqref{eq:complete-case-estimator}—which places us within an $O_p(n^{-1/2})$ neighborhood of the true minimizer $\theta_\star$. A single Newton update from this starting point then yields the IPI point estimate:

\begin{definition}[IPI point estimate]
\begin{equation}
    \hat{\theta}_{\text{\em IPI}, \lambda}
    := \hat{\theta}_n - \hat{H}_n^{-1} \nabla L_{\text{\em IPI}}(\hat{\theta}_n; \lambda),
\end{equation}
where the complete-case estimator is given by:
\begin{equation}\label{eq:complete-case-estimator}
    \hat{\theta}_n := \argmin\nolimits_\theta \mathbb{P}_n \left[\ell(X; \theta)\right],
\end{equation}
and $\hat{H}_{n}$ is a consistent estimator for $\mathbb{E}[\nabla^2 L_{\IPI}(\theta_\star; \lambda)]$. Under MCAR, an option is $\hat{H}_n := \mathbb{P}_n [\nabla^2 \ell(X; \hat{\theta}_n)]$, the sample Hessian on the fully observed samples. 

\end{definition}

\begin{theorem}[IPI central limit theorem under MCAR]\label{thm:ipi_clt_mcar} Given MCAR and regularity assumptions in Appendix \ref{app:IPI}, if $\hat{\lambda} = \lambda + o_p(1)$ for some fixed $\lambda$,
\begin{align*}
\sqrt{n}\left(\hat{\theta}_{\text{\em IPI}, \hat{\lambda}} - \theta_\star\right) \overset{d}{\rightarrow} \mathcal{N}\left(0, \Sigma\right),
\end{align*}
where
\begin{equation}\label{eq:IPI_asymp_var}
    \Sigma = H_{\theta_\star}^{-1}V_{\theta_\star, \lambda} H_{\theta_\star}^{-1},
\end{equation}
in which $ H_{\theta_\star} = \mathbb{E}\left[\nabla^2 \ell(X; \theta_*)\right]$ is the population Hessian and
\begin{align}\label{eq:ipi_loss_asymp_var}
    V_{\theta_\star, \lambda}
    &:= \operatorname{Var} \left(\nabla \ell(X; {\theta_\star}) - \frac{1}{R}\sum_{r=1}^R \lambda_r \nabla \ell(f(O\circ m_r); {\theta_\star})\Bigg|M = \bm{1}_d\right)\\
    &\hspace{5mm} + \sum_{r=1}^R \left(\frac{\lambda_r}{R}\right)^2 \frac{p_0}{p_r} \Var\left(\nabla \ell(f(O \circ m_r); {\theta_\star})\mid M = m_r\right). \nonumber
\end{align}
\end{theorem}

\begin{corollary}[IPI valid inference under MCAR]
\label{cor:ipi_clt_mcar}
    In the setting of Theorem \ref{thm:ipi_clt_mcar}, the confidence set 
    \[
        \mathcal{C}_{\IPI, \alpha}  = \left\{\theta: \left\|\Sigma^{-1/2} \left(\hat{\theta}_{\IPI, \hat{\lambda}} - \theta\right)\right\|_2^2 \leq \frac{\chi^2_{p, 1-\alpha}}{n}\right\}
    \]
    has $\lim_{n \rightarrow \infty} \mathbb{P}(\theta_\star \in \mathcal{C}_{\IPI, \alpha}) = 1-\alpha$ and is thus asymptotically valid. 
\end{corollary}

Algorithm \ref{alg:ipi} provides steps to constructing a confidence interval for scalar $\theta_\star$. The reader will find proofs of both Theorem \ref{thm:ipi_clt_mcar} and Corollary \ref{cor:ipi_clt_mcar}
in Appendix \ref{app:IPI}.

\begin{algorithm}[t]
\caption{IPI confidence interval (CI) for scalar $\theta_\star$ under MCAR}
\label{alg:ipi}
\begin{algorithmic}[1]
\Require
    Observations $\{O_i=((X_i,U_i)\circ M_i)\}_{i=1}^{N}$;
    Imputation model $f$;
    Error level $\alpha\!\in(0,1)$
\Ensure $(1-\alpha)$ two–sided CI for the target $\theta_\star$\; 

\smallskip
\State \textbf{Initialize at complete-case estimate:}
       $\displaystyle
       \hat\theta_n \gets \argmin\nolimits_{\theta}\;
       \mathbb P_n\!\bigl[\ell(X;\theta)\bigr]$
\State \textbf{Hessian estimate:}
       $\displaystyle
       \hat H_n \gets
       \mathbb P_n\!\bigl[\nabla^2\ell(X;\hat\theta_n)\bigr]$
\State \textbf{Variance estimate:} Calculate $\hat{\sigma}^2$ from Eqs. \eqref{eq:IPI_asymp_var} and \eqref{eq:ipi_loss_asymp_var} with sample covariance, $\hat{H}_n$, and $\hat{\theta}_n$ as plug-in estimates

\State \textbf{Weights:} minimize the plug-in asymptotic variance over 
       $\hat{\lambda}=(\hat\lambda_1,\dots,\hat\lambda_R)$ 
\State $\displaystyle
       \nabla L_\IPI \gets
       \mathbb P_n\!\bigl[\nabla\ell(X;\hat\theta_n)\bigr]$
\For{$r \gets 1$ \textbf{to} $R$} 
    \State $\displaystyle
        \nabla L_\IPI \gets \nabla L_{\IPI}
        + \hat\lambda_r
          \Bigl\{
            \mathbb P_{\tilde N_r}\!\bigl[\nabla\ell(f(O);\hat\theta_n)\bigr]
            - \mathbb P_n\!\bigl[\nabla\ell(f(O\circ m_r); \hat\theta_n)\bigr]
          \Bigr\}$
\EndFor

\State \textbf{Calculate one-step estimator:} $\hat{\theta}_\IPI \gets \hat{\theta}_n - \hat{H}_n^{-1} \nabla L_\IPI$

\State \textbf{Critical value:} $z_{1-\alpha/2}$ (standard normal)
\State \Return
       $\displaystyle
       \bigl[
         \hat{\theta}_\IPI
         \;\pm\;
         z_{1-\alpha/2}\,
        \widehat{\sigma}/ \sqrt{n}
       \bigr]$
\end{algorithmic}
\end{algorithm}

\subsection{CIPI under MCAR}
IPI assumes access to a pre-trained imputation model $f$, but in practice the imputer is typically learned from the data itself. To avoid the sample inefficiency of a single train-test split---without imposing stronger assumptions (e.g., Donsker) on the imputer---we use a common technique known as \textit{cross-fitting} (also leveraged in CPPI \cite{zrnic_cppi}): partition the data into $K$ folds, train an imputer $f^{(k)}$ on all but the $k$th fold, and apply it only to that fold. This way, all observations contribute to both training and inference, yielding more stable and efficient estimates.

Unlike CPPI, however, we assume the imputation model can be trained on partially observed data. As a result, we split the \emph{entire} dataset---including partially observed samples---into $K$ folds. This increases the training sample size from $n(K-1)/K$ to $N(K-1)/K$, which can yield substantial gains when partially observed data far outnumber fully observed data ($\tilde{N} \gg n$).

\begin{definition}[CIPI loss]
Let $I_k$ denote the $k$th fold of the data and let $L_{\IPI}^{(k)}(\theta; \lambda)$ denote $L_{\IPI}(\theta; \lambda)$ calculated on the $k$th fold with imputation algorithm $f^{(k)}$. We define the CIPI loss as:
\begin{align*}
    L_{\CIPI}(\theta; \lambda) := \frac{1}{K}\sum\nolimits_{k =1}^K  L_{\IPI}^{(k)}(\theta; \lambda).
\end{align*}    
\end{definition}

\begin{definition}[CIPI point estimate] For any given $\lambda\in \mathbb{R}^R$, the corresponding one-step point estimate then simply replaces the IPI loss with the CIPI loss:
\begin{align*}
    \hat{\theta}_{\CIPI, \lambda} := \hat{\theta}_n - \hat{H}_n^{-1} \nabla L_{\CIPI}(\hat{\theta}_n; \lambda).
\end{align*}
\end{definition}

Constructing valid confidence intervals with this estimate requires an extra stability assumption – much like that in  CPPI \cite{zrnic_cppi}. In particular, for each realized observation $o \in (\mathbb{R} \cup {\texttt{NA}})^d$, let $\bar{f}(o) := \mathbb{E}_{\text{train}}[f^{(1)}(o)]$ denote the \textit{average imputer} over training sets of size $N(K-1)/K$. The stability requirement (Assumption \ref{assump:stability assumptions}) requires the gradient loss of any realized $f^{(1)}$ to converge to that of $\bar{f}$ across missing patterns. Note that this does not require $f$ to be well-specified—only stable. 
\begin{assumption}[Stability requirements]\label{assump:stability assumptions}
    For each $r \in [R]$, $\tilde{r} \in \{0, r\}$, and for each $\theta$,
    \begin{align*}
        \sqrt{K\Var \left(\nabla \ell(f^{(1)}(O \circ m_r); \theta) - \nabla {\ell}(\bar{f}(O\circ m_r); \theta)|M = m_{\tilde{r}}, f^{(1)}\right)} \overset{L^1}{\rightarrow} 0.
    \end{align*}
\end{assumption}

With Assumption \ref{assump:stability assumptions}, we get asymptotic normality as in Theorem \ref{thm:ipi_clt_mcar} –  please see Theorem \ref{thm:cipi_clt_mcar} in Appendix \ref{app:CIPI}. The following asymptotic guarantees also hold:

\begin{corollary}[CIPI valid inference under MCAR]\label{cor:cipi_ci}
    The confidence set 
    \[
        \mathcal{C}_{\CIPI, \alpha}  = \left\{\theta: \left\|\bar{\Sigma}^{-1/2} \left(\hat{\theta}_{\CIPI, \hat{\lambda}} - \theta\right)\right\|_2^2 \leq \frac{\chi^2_{p, 1-\alpha}}{n}\right\}
    \]
    has $\lim_{n \rightarrow \infty} \mathbb{P}(\theta_\star \in \mathcal{C}_{\CIPI, \alpha}) = 1-\alpha$ and is thus asymptotically valid.
\end{corollary}

\textbf{Bootstrap variance estimation.} Masked in the above statement is the dependence of the asymptotic variance $\bar{\Sigma}$ on the true average imputer $\bar{f}$. Since $\bar{f}$ is not available in practice, we follow the bootstrap-based approach introduced in CPPI \cite{zrnic_cppi} to estimate the variance; further details are provided in Appendix~\ref{app:boot_var_est}. This procedure requires retraining the imputation algorithm at least 50--100 times, which can be computationally expensive. Developing more efficient methods for estimating this variance remains an interesting direction for future research aimed at improving the computational scalability of our framework.

\subsection{Power tuning}\label{sec:power-tuning}
Following \citet{angelopoulos2024ppi++}, we refer to $\lambda \in  \mathbb{R}^R$ 
% ${\ejc{XXX}}$ \ejc{This is to remind ourselves that $\lambda$ is a vector.} 
as a power-tuning parameter, with the pooled estimate $\hat{\theta}_{\text{IPI pooled}}$ using $\lambda_r$ proportional to the missing pattern frequency. We note that the pooled estimate's performance may be poor when the most common patterns are the ones on which $f$ imputes poorly. Here, we leverage similar ideas as discussed in the power-tuning literature \cite{angelopoulos2024ppi++, ji2025predictionssurrogatesrevisitingsurrogate_reppi, kluger2025predictionpoweredinferenceimputedcovariates, angelopoulos2024ppi++, miao2024assumptionleandataadaptivepostpredictioninference_pspa, gan2024prediction, gronsbell2024another} to choose $\lambda$ to ensure sample efficiency never falls below the complete-case baseline under MCAR, even in the case where the ML model is poorly aligned with the missing data imputation problem and only adds noise.

In particular, we seek $\lambda_\star$ minimizing $\Sigma$, the asymptotic variance of the IPI estimate $\theta_{\IPI}$:
\begin{equation}
    \lambda_\star := \argmin\nolimits_{\lambda} \text{Tr}({\Sigma}),
\end{equation}
or if we are interested in just one coordinate, say  $\theta_{\star, j}$, we can use $ \lambda_\star := \argmin\nolimits_{\lambda} {\Sigma}_{jj}$.

Because $\Sigma$ is unknown in practice, we approximate it using the sample variances and covariances evaluated at the complete-case estimator $\hat{\theta}_n$. We denote the minimizer of the resulting objective by $\hat{\lambda}_\star$.

We provide closed forms for both $\lambda_\star$ and $\hat{\lambda}_\star$ in Appendix \ref{app:Tuning}.

Note that this choice of $\lambda$ guarantees improved sample efficiency and therefore tighter confidence intervals compared to other valid approaches under MCAR. These alternatives are the complete-case estimator $\hat{\theta}_n$ (discarding samples with missing data), and PPI/IPI applied to a single missing pattern. Results for the tuned versus pooled estimators for both IPI and CIPI under MCAR settings are provided in Section \ref{expt:mcar}.

\section{IPI and CIPI under more general missing mechanisms}\label{sec:moving_away_mcar}
The MCAR assumption is often too strong to hold in practice. While most of the PPI literature implicitly adopts this assumption—aside from a few exceptions that assume full knowledge of the missingness mechanism \cite{ppi, kluger2025predictionpoweredinferenceimputedcovariates}—we examine how our method performs under more general forms of missingness. This contrasts with prior PPI-related and semiparametric approaches, which typically adapt their methods to AIPW-style estimators with either known or correctly specified propensity scores \cite{kluger2025predictionpoweredinferenceimputedcovariates, tsiatis2006semiparametric}. Rather than modifying our method to accommodate these assumptions, we take a different viewpoint: we investigate what remains identifiable under non-MCAR missingness settings without relying on untestable assumptions or potentially misspecified propensity scores.

A key ingredient for the validity of IPI is the ability to debias each missingness pattern using the subset of fully observed data. In general missing data settings, valid debiasing remains possible if, for each missingness pattern, the average imputation bias can be approximated by the average behavior of the imputer applied to masked complete data. This observation motivates a relaxation of MCAR that we term the first-moment MCAR assumption.

\begin{assumption}[MCAR first moment assumption for imputation algorithm $f$]\label{assump:MCAR_first_moment} For each missingness pattern $r \in [R]$,
\begin{equation}
    \mathbb{E}[\nabla \ell(f(O); \tilde{\theta}_\star) \mid M = m_r]
         = \mathbb{E}[\nabla \ell(f(O\circ m_r); \tilde{\theta}_\star) \mid M = \bm{1}_d].
\end{equation}
\end{assumption}
Note that this is a strict relaxation of the MCAR assumption, which requires the \textit{distributions} for our underlying distribution $D$ to be independent of observed missing pattern. Instead, we just require that the imputed score functions are equal on average.

Provided Assumption \ref{assump:MCAR_first_moment} holds, under general missing mechanisms, IPI can still provide valid inference for the subpopulation with $D$ fully observed, shifting the estimand to
\begin{equation}\label{eq:shifted_target}
    \tilde{\theta}\star := \argmin_\theta \mathbb{E}\left[\ell(X; \theta) \mid M = \bm{1}_d\right].
\end{equation}
This shift is natural in applications where the focus lies on high-risk or clinically flagged groups—such as patients selected for additional testing—rather than the entire population. Researchers in these scenarios often combine focused studies with broader historical datasets, which introduces complex, nonrandom missingness that IPI is designed to handle.

Asymptotic results similar to Theorems \ref{thm:ipi_clt_mcar} and \ref{thm:cipi_clt_mcar} hold when replacing MCAR and $\theta_\star$ with Assumption \ref{assump:MCAR_first_moment} (for certain choices of $f$) and $\tilde{\theta}_\star$, respectively (see Appendix \ref{app:away-mcar}, Theorems \ref{thm:general_ipi} and \ref{thm:general-cipi} for full statements and proofs).

\subsection{Diagnostics}
Since Assumption \ref{assump:MCAR_first_moment} is simply a requirement on the equivalence of estimable moments, it is asymptotically testable (e.g., for the semi-supervised setting, this is a single $p$-dimensional z-test). A direct diagnostic for blockwise missingness involves the test statistic
\begin{align*}
    \hat{T}_\mathrm{full} := \begin{bmatrix}
        \mathbb{P}_{\tilde{N}_1}[\nabla \ell (f(O); \hat{\theta}_n)] - \mathbb{P}_{n}[\nabla \ell (f(O \circ m_1); \hat{\theta}_n)]\\
        \vdots \\
         \mathbb{P}_{\tilde{N}_R}[\nabla \ell (f(O); \hat{\theta}_n)] - \mathbb{P}_{n}[\nabla \ell (f(O \circ m_R); \hat{\theta}_n)]
    \end{bmatrix} \in \mathbb{R}^{pR},
\end{align*}
but its dimensionality scales with both $p$ and $R$. This motivates a simpler, lower-dimensional diagnostic:
\begin{align*}
    \hat{T}_\IPI := \frac{1}{R}\sum_{r =1}^R \hat{\lambda}_{\star, r} (\mathbb{P}_{\tilde{N}_r}[\nabla \ell (f(O); \hat{\theta}_n)] - \mathbb{P}_{n}[\nabla \ell (f(O \circ m_r); \hat{\theta}_n)]) \in \mathbb{R}^{p}.
\end{align*}

Proposition \ref{prop:asymp_behav_ipi_test} ensures that testing for large values of $\|\hat{V}_T^{-1/2} \hat{T}_{\IPI}\|_2^2$ yields an asymptotically valid level-$\alpha$ test, where $\hat{V}_T^{-1/2}$ is consistent for $V_T$, and an asymptotically uniformly-distributed p-value under Assumption \ref{assump:MCAR_first_moment}.

\begin{proposition}[Asymptotic behavior of test statistic $\hat{T}_\IPI$]
\label{prop:asymp_behav_ipi_test}
Under Assumption \ref{assump:MCAR_first_moment} and regularity assumptions in Appendix \ref{app:IPI}, we have
\begin{align*}
    \sqrt{n} \hat{T}_\IPI \overset{d}{\rightarrow} \mathcal{N} \left(0, V_T \right)
\end{align*}
where 
\begin{multline*}
    V_T := \Var \left(\frac{1}{R}\sum_{r = 1}^R \lambda_r \left((\tilde{H}_{\tilde{\theta}_\star}^r - {H}_{\tilde{\theta}_\star}^r) {H}_{\tilde{\theta}_\star}^{-1}\nabla \ell(X; \tilde{\theta}_\star) +  \nabla \ell(f(O\circ m_r); \theta_\star)\right) \bigg| M = \bm{1}_d\right) \\
     + \sum_{r =1}^R \frac{p_0}{p_r}\left(\frac{\lambda_r}{R}\right)^2 \Var\left(\nabla \ell(f(O\circ m_r); \theta_\star) | M = m_r\right)
\end{multline*}
with 
\begin{align*}
    H_{\tilde{\theta}_\star} &:= \mathbb{E}\left[\nabla^2 \ell(X; \tilde{\theta}_\star)\big| M = \bm{1}_d\right],\\
    H_{\tilde{\theta}_\star}^r &:= \mathbb{E}\left[\nabla^2 \ell(f(O \circ m_r); \tilde{\theta}_\star)\big| M = \bm{1}_d\right], \text{ and }\\
    \tilde{H}_{\tilde{\theta}_\star}^r &:= \mathbb{E}\left[\nabla^2 \ell(f(O); \tilde{\theta}_\star)\big| M = m_r\right].
\end{align*}
\end{proposition}

We provide a proof of the above proposition as well as a more detailed empirical comparison of $\hat{T}_\mathrm{full}$ with $\hat{T}_\IPI$ in Appendix \ref{app:diagnostics} and \ref{app:expt-diagnostics}, respectively.

\section{Experiments}\label{sec: expts}

In this section, we present two main findings. (i) In controlled MCAR settings, IPI and CIPI exploit multiple missingness patterns to achieve substantial sample-efficiency gains over single-pattern PPI and complete-case estimation, while maintaining nominal coverage; by contrast, naive single-imputation and AIPW baselines are biased or unstable and exhibit undercoverage. (ii) In a clinical dataset that departs from MCAR, we obtain reasonable inference for the shifted target $\tilde{\theta}_\star$ and our diagnostics detect the resulting distributional shift.

Throughout, we focus on regimes with a small number of missingness patterns $R$ relative to the sample sizes $n$ and $\tilde{N}_r$, as is common when merging multi-modal datasets that induce blockwise missingness. For each $n$, we define a method’s effective sample size as the number of complete-case samples required to match its confidence-interval width: \[N_\text{eff} = n ({\text{complete-case baseline interval width}}/{\text{new method interval width}})^2.\] 

Additional experimental details, formal baseline definitions, and supplemental results are shown in Appendix \ref{app:exptl_details}.

\subsection{Proof-of-concept experiments under MCAR }\label{expt:mcar}
We first evaluate IPI and CIPI under controlled MCAR simulations against up to four baselines: (i) complete-case analysis; (ii) single-pattern IPI (equivalently PPI in the single-pattern setting); (iii) naïve single-imputation (impute with a specified method, then analyze as if fully observed); and, when applicable, (iv) an AIPW estimator (details in Appendix \ref{app:aipw_expts}). 

Because the true parameter $\theta_\star$ is known in these simulated settings, we can directly assess bias and coverage. As shown in the following experiments, naive single-imputation yields biased estimates, and the AIPW baseline is unstable, both leading to undercoverage. By contrast, complete-case and IPI-based methods achieve nominal coverage; moreover, tuned IPI and CIPI consistently deliver substantial sample-efficiency gains over the valid baselines (complete-case and single-pattern IPI).

We study two MCAR settings—synthetic data (\textbf{factor model}, $n=200$) and real data with simulated missingness (\textbf{census}, $n=2000$)—averaging results over $1{,}000$ trials. We vary the ratio of partially to fully observed samples from 10 to 50 to examine how additional partial information affects inference. IPI uses a 10\% train–test split, and CIPI is implemented with $K=10$-fold cross-fitting. All results are reported for 90\% confidence intervals.

\textbf{Experiment 1: factor model.}
To evaluate our method in a controlled setting, we simulate data from a factor model: $X \sim \mathcal{N}(0, \Sigma_X) \in \mathbb{R}^{20}$ with $\Sigma_X = FF^\top + \sigma^2 I$, where $\sigma^2$ is chosen so that $q = 2$ latent factors explain 50\% of the variance (introducing heterogeneity that makes imputation nontrivial). 

The factor matrix $F$ and idiosyncratic variance $\sigma^2$ are unknown. Missingness is generated by independently masking each feature with probability 0.2, then retaining ten resulting patterns with equal probability, yielding an average missing rate of 25\% with missingness patterns in Figure \ref{fig:fm-results}(a).

For IPI-based methods, we adapt the EM algorithm, which typically performs in-place imputation \cite{dempster1977maximum}. Specifically, we first estimate the covariance structure from a held-out training set and then impute missing entries using conditional means based on the learned covariance. This procedure ensures that the imputation function $f$ is independent of the test data used to construct the prediction interval—a key requirement for the validity of general PPI methodology \cite{ppi}.

Our target parameter is the regression coefficient of $X_0$ in a linear regression of $X_2$ on $(X_0, X_1)$. As seen in Figure \ref{fig:fm-results}(a), we observe multiple missing patterns in these features with at least one of these features missing for almost all the dataset. 

Baselines include: (1) complete-case; (2) IPI using the best single-pattern; (3) AIPW; and (4) naive EM single imputation without bias correction. Because the latter two methods do not achieve nominal 
90\% coverage, we omit effective sample-size gains for them from the main text. See Appendix \ref{app:exptl_details} Figure \ref{fig:fm_ipi_classical_full_comparison} for the full comparison including five different imputation strategies and additional results, where the IPI framework both recovers nominal coverage when naive imputation methods fail and consistently substantially improves sample efficiency over complete-case analysis.

\begin{figure}[h!]
    \centering
    \includegraphics[width=\linewidth]{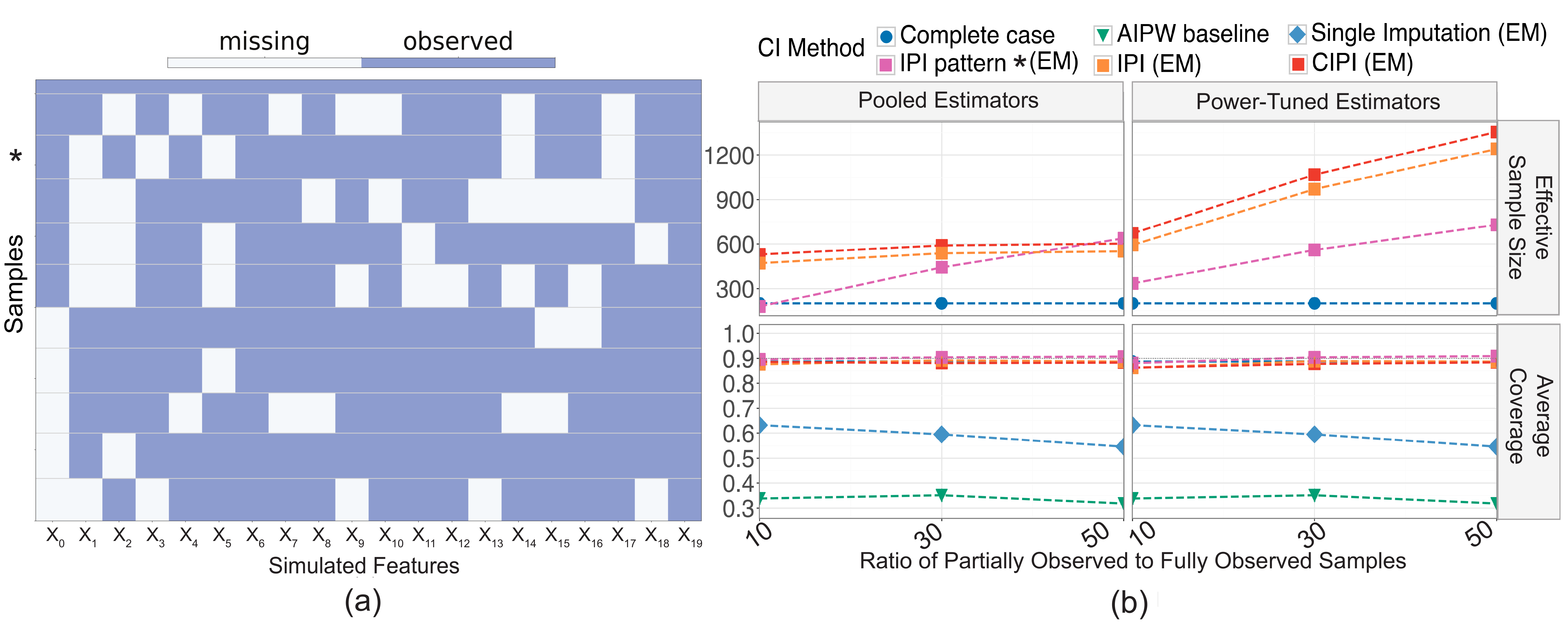}
    \caption{\textbf{Results for the factor model.} (a) Missing-data patterns from a representative run for $\tilde{N}/n = 30$. Each horizontal block is a distinct pattern, with row height proportional to its frequency; the star marks the top-performing pattern. (b) Average coverage and effective sample size across four methods: complete-case, best single-pattern imputation, IPI, and CIPI. Average coverage is also reported for other baselines that do not achieve nominal coverage.}
    \label{fig:fm-results}
\end{figure}

Figure~\ref{fig:fm-results} (b) summarizes the experimental results. Complete-case and IPI-based methods achieve nominal coverage (bottom row) while other baselines fail to achieve this coverage. Comparing effective sample sizes across pooled and power-tuned estimators (top row) demonstrates (i) power-tuning with $\hat{\lambda}_\star$ is crucial for effectively aggregating information across multiple missing patterns and (ii) CIPI proves to be more sample efficient than IPI by exploiting cross-fitting to improve efficiency over train-test splitting.

\textbf{Experiment 2: studying schooling level and income in census surveys.} Survey data often exhibit complex, multi-pattern missingness. We evaluate IPI on the American Community Survey (ACS) Public Use Microdata Sample (PUMS) \cite{census_survey_folktables}, a representative setting with diverse missing patterns. Our target is the coefficient of schooling level in a linear regression of personal income on schooling and age. We include 14 auxiliary demographic features (e.g., sex, marital status, disability, parental employment) to boost imputation. For blackbox imputation, we use a version of MissForest \cite{stekhoven2012missforest}, which performs iterative imputation with gradient-boosted trees \cite{ke2017lightgbm} to handle mixed data types.

The ACS data exhibit structured missingness due to non-response and privacy-related masking. To simulate missingness under the MCAR assumption–while preserving the marginal distributions of both the data and the missingness patterns–we first estimate empirical patterns from partially observed records and apply them, with their observed frequencies, to a fully observed subset. This results in a 20\% overall missing rate while maintaining realistic marginal structure. This synthetic masking preserves access to ground truth values, enabling rigorous coverage evaluation. We treat the fully observed finite population as the true distribution, with the corresponding regression coefficient taken as the true $\theta_\star$. Coverage is assessed by whether this finite-population value falls within the estimated interval.

Baselines include: (1) complete-case analysis; (2) the best single-pattern IPI method; and (3) naive MissForest imputation without any bias correction. Due to the low coverage observed for the latter two baselines, we do not report their effective sample size gains. The AIPW baseline is omitted, as it is not well-suited to these complex data settings, where - for instance - one of the features has over 100 distinct categories.

Figure~\ref{fig:acs-results}(b) shows that IPI remains robust in this more complex setting, achieving greater effective sample size and maintaining valid confidence interval coverage while naive single imputation with the same method has less than 20\% average coverage (bottom row). Power-tuning further amplifies this gain, boosting the effective sample size from 2$\times$ to 4$\times$ (top row). 

\begin{figure}[h!]
    \centering
    \includegraphics[width=\linewidth]{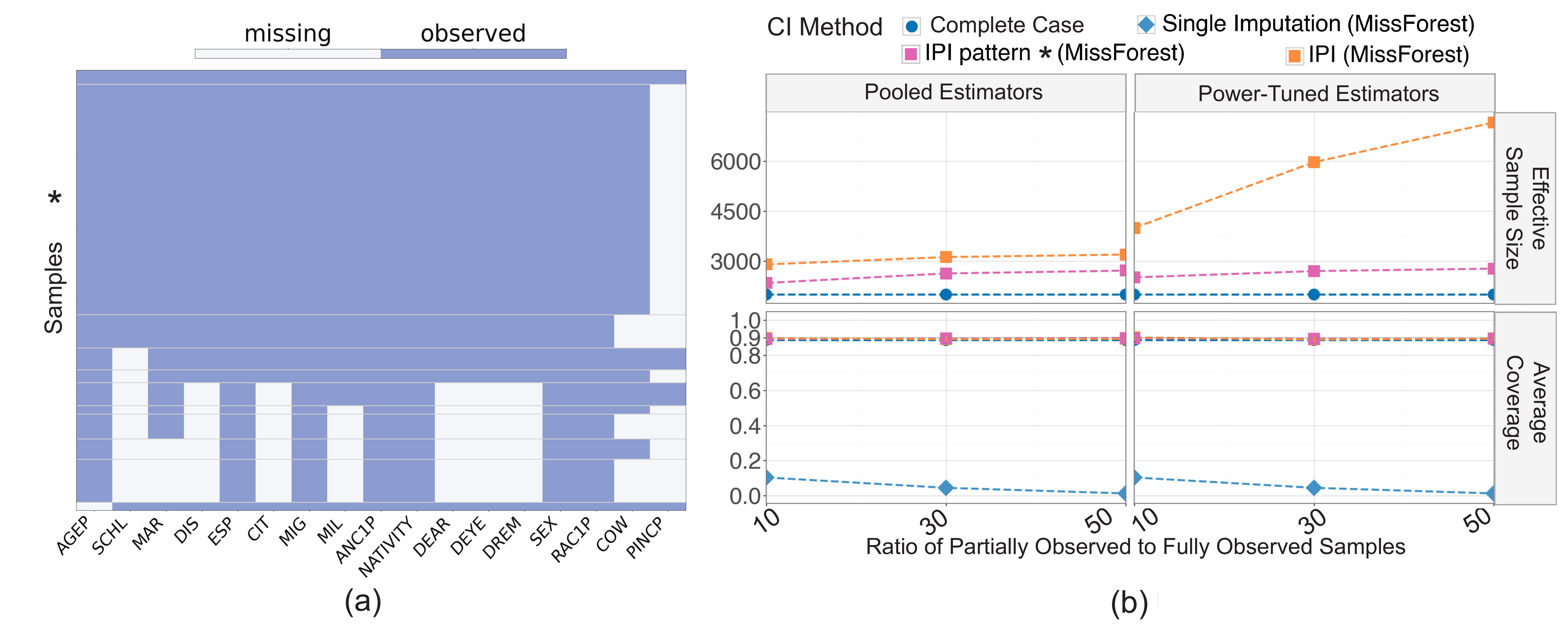}
    \caption{\textbf{Results for census.} (a) 
    Missing-data patterns from a representative run for $\tilde{N}/n = 30$. Each horizontal block is a distinct pattern, with row height proportional to its frequency; star marks the top-performing pattern. 
    Feature codes are provided in Appendix \ref{app:exptl_details}.  (b) Average coverage and effective sample size across three methods: complete-case, best single-pattern imputation, and IPI. Average coverage is also reported for other baselines that do not achieve nominal coverage.}
    \label{fig:acs-results}
\end{figure}

\subsection{Real data: allergen dataset}\label{sec:allergen_expt}
In clinical datasets, missing data often occur due to patients skipping certain demographic questions or lab results being unavailable for certain subjects. These settings are not only practically important but also highlight scenarios where leveraging informative machine learning imputations can be a cost-effective way to boost effective sample size and meaningfully improve inference. We study one such significant example: the allergen chip challenge dataset \cite{allergen_chip}, which combines allergen-specific IgE chip readings with clinical and demographic data—where missingness primarily affects the latter two modalities. After preprocessing the allergen dataset, we obtain $N = 1760$ samples, $d = 30$ features, $n = 429$ fully observed samples across all $30$ features, and 6 distinct missing patterns. Imputation is performed using MissForest.

We perform inference on the shifted target $\tilde{\theta}_\star$: the coefficient of a feature in the linear regression of an IgE response on age, sex and said feature for the fully observed population. For the IgE responses, we focus on five separate regressions for the peanut allergens Ara h 1, Ara h 2, Ara h 3, Ara h 6, and Ara h 8; and for the features, experiment Figure \ref{fig:allergen-dataset}(a) depicts coefficient confidence intervals when the feature of interest is rhinitis severity (ranging from 0 to 5) and experiment Figure \ref{fig:allergen-dataset}(b) depicts results when the feature of interest is asthma severity (ranging from 0 to 5).

Figure \ref{fig:allergen-dataset}(a) (rhinitis severity) shows that CIPI consistently shortens CIs relative to complete-case analysis, with IPI improving analysis for a few regressions including that for Ara h 1. In this regression, the CI is visibly shorter and the diagnostic is large (Figure \ref{fig:allergen-dataset}(c), Expt (a) Ara h 1 has $p=0.919$), indicating no strong evidence of the violation of the MCAR first moment assumption and that IPI/CIPI is likely performed under the stated null.

In contrast, Figure \ref{fig:allergen-dataset}(b) (asthma severity) exhibits more disagreement — most notably for Ara h 2, Ara h 3, and Ara h 6—mirrored by the smaller diagnostic p-values in panel c (Expt (b)). These warnings suggest that, for asthma severity, the combination of MissForest and linear adjustment may be extrapolating in regions poorly supported by the fully observed data, even when CIs appear narrow.

Together, the figure highlights two desirable properties for the IPI framework and diagnostics: (i) CIPI can boost precision over complete-case analysis by leveraging information from partially observed units; and (ii) unlike prior PPI workflows that offer no guardrails, our diagnostic alerts users when imputation-assisted gains may come at the expense of validity.

\begin{figure}[h!]
    \centering    \includegraphics[width=1\linewidth]{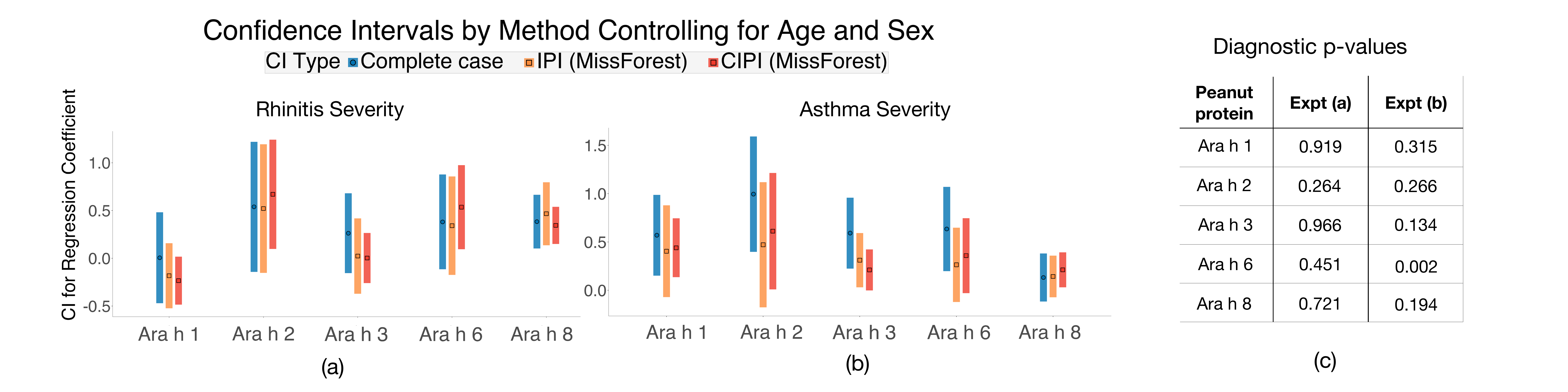}
    \caption{\textbf{Allergen chip dataset inference}. (a) Confidence intervals when regressing each peanut protein on age, sex, and rhinitis severity. In plot (b), regressing instead on age, sex, and asthma severity results in more disagreement between the IPI-based and complete-case methods. This shift is echoed by the diagnostic p-values (c).
    }
    \label{fig:allergen-dataset}
\end{figure}
\section{Discussion}
We introduced IPI and CIPI, extending prediction-powered inference to accommodate multiple missing-data patterns through blackbox imputation. By leveraging a fully observed subset to debias imputed samples, our methods substantially boost sample efficiency while preserving valid inference under MCAR and under a relaxed moment condition that extends beyond MCAR. We demonstrate these gains empirically, including on a real clinical dataset, and propose a diagnostic to assess when such debiasing is valid. Future work could seek to reduce the reliance on fully observed auxiliary features $U$ in order to develop even more sample-efficient and valid methods. When missing patterns each have few observations, clustering patterns or adding regularization to the power-tuning weights may improve stability of this framework. Our results highlight that principled debiasing enables machine-learning imputations to support valid and powerful inference, even in complex and heterogeneous data settings.

{\small
\bibliographystyle{unsrtnat}
\bibliography{references}
}

\appendix
\include{appendix}

\end{document}

%% file: appendix.tex
\section{Connection with AIPW estimators}\label{app:semiparam}

In this section, we compare the IPI method to the broader class of augmented inverse probability weighting (AIPW) estimators for non-monotone missingness in semiparametric theory \cite{robins1994, tsiatis2006semiparametric}. In particular, for ease of comparison, we focus on estimating the population linear regression coefficient $\theta_\star$, which satisfies the moment condition $\mathbb{E}[X(Y - X^\top \theta)] = 0$.

A key feature of our approach is that we do not assume that the restricted moment model $\mathbb{E}[Y|X] = X^\top \theta$ holds. This is a distinction from the methods discussed in \citet{robins1994, sun2018inverse, sun2018inverse2}, where this assumption is leveraged to further optimize estimator efficiency. Our method instead simply targets the coefficient of the best linear predictor for $Y$. Thus, to maintain consistency of the estimate, the AIPW baseline starts with  the basic inverse probability weighting (IPW) estimator (equivalent to the complete-case estimator under MCAR) that is the solution to $ \frac{1}{N}\sum_{i =1 }^N \frac{\mathbf{1}(M_i = \mathbf{1}_d)}{\hat{\mathbb{P}}[M = \mathbf{1}_d]}X_i(Y_i - \theta^\top X_i) =\mathbb{P}_n [X(Y- \theta^\top X)] = 0$, and seeks to improve its precision by projecting it onto a space of mean-zero functions of the observed data – in particular, the augmentation space $\mathcal{A} := \left\{\phi: E[\phi(O, M)|D] = 0\right\}$. 

The most efficient estimator is the solution to 
\begin{align*}
    \mathbb{P}_N \left[\frac{\mathbf{1}(M = \mathbf{1}_d)}{\mathbb{P}[M = \mathbf{1}_d]} C_1^{\mathrm{opt}} X(Y - \theta^\top X) - \mathcal{P}_\mathcal{A}\left( \frac{\mathbf{1}(M = \mathbf{1}_d)}{\mathbb{P}[M = \mathbf{1}_d]}C_1^{\mathrm{opt}} X(Y - \theta^\top X)\right)\right] = 0,
\end{align*}
where  $\mathcal{P}_\mathcal{A}$ is the $L^2$ projection of the IPW estimator onto the augmentation space $\mathcal{A}$, and $C_1^{\mathrm{opt}}\in\mathbb{R}^{p\times p}$ is a variance-minimizing preconditioner (it is irrelevant without augmentation, but matters once you project).

As noted in \citet{robins1994} and \citet{tsiatis2006semiparametric}, this projection lacks a closed-form in the general non-monotone case (Sec. 7.2 \cite{robins1994}), even when the full data distribution is known. The resulting estimator construction, further detailed in Tsiatis \cite{tsiatis2006semiparametric}, is theoretically sound but often computationally demanding and unstable in practice.

A common, more tractable alternative is to project the IPW estimator onto a well-chosen, finite-dimensional linear subspace of $\mathcal{A}$. For the AIPW baseline method, we follow the common approach \cite{tsiatis2006semiparametric, sun2018inverse, sun2018inverse2} to project onto a finite-dimensional linear subspace spanned by low-degree monomials of the observed features. Define
\begin{align*}
    \mathcal{A}_\mathrm{AIPW}(O, M) := \Bigg\{C_2 \cdot A_\mathrm{AIPW}(O, M):C_2 \in \mathbb{R}^{p \times \sum_{r = 1}^R \left(1 + d_r + {d_r \choose 2}\right)} \Bigg\},
\end{align*}
where $d_r\leq d$ is the number of observed features in missingness pattern $m_r$, and 
\begin{equation}\label{app-eq:aipw_class1_augmentation}
\begin{split}
    A_{\mathrm{AIPW}}(O, M) &:= 
    \Bigg[
    \left(\frac{\mathbf{1}\{M = \mathbf{1}_d\}}{\mathbb{P}[M = \mathbf{1}_d]} - \frac{\mathbf{1}\{M = m_1\}}{\mathbb{P}[M = m_1]}\right) V_1\\
    &\hspace{5mm}\left(\frac{\mathbf{1}\{M = \mathbf{1}_d\}}{\mathbb{P}[M = \mathbf{1}_d]} - \frac{\mathbf{1}\{M = m_2\}}{\mathbb{P}[M = m_2]}\right) V_2\\
    &\hspace{30mm} \vdots \\
    &\hspace{5mm} \left(\frac{\mathbf{1}\{M = \mathbf{1}_d\}}{\mathbb{P}[M = \mathbf{1}_d]} - \frac{\mathbf{1}\{M = m_R\}}{\mathbb{P}[M = m_R]}\right) V_R\Bigg] \in \mathbb{R}^{\sum_{r = 1}^R \left(1 + d_r + {d_r \choose 2}\right)},
\end{split}
\end{equation}
with $V_r$ consisting of monomials up to degree two in the components observed under $m_r$. For example, if $D = (X_1, X_2, U_1)$ and $m_r = (1, \texttt{NA}, 1)$, i.e., only $X_1$ and $U_1$ are observed, then $V_r := (1\hspace{2mm} X_1\hspace{2mm}  U_1\hspace{2mm} X_1^2\hspace{2mm}  X_1U_1\hspace{2mm}  U_1^2)$.

The corresponding AIPW-baseline is the one-step estimator for the solution to 
\begin{align*}
    \mathbb{P}_N \left[\frac{\mathbf{1}(M = \mathbf{1}_d)}{\mathbb{P}[M = \mathbf{1}_d]}C_1^{\mathrm{opt, AIPW}}X(Y - \theta^\top X) - \mathcal{P}_\mathcal{A_\mathrm{AIPW}}\left(\frac{\mathbf{1}(M = \mathbf{1}_d)}{\mathbb{P}[M = \mathbf{1}_d]}C_1^{\mathrm{opt, AIPW}}X(Y - \theta^\top X)\right)\right] = 0,
\end{align*}
where $C_1^{\mathrm{opt, AIPW}} \in \mathbb{R}^{p\times p}$ is the matrix optimizing the resulting variance.

However, even in high-dimensional settings and more complex data settings, this more tractable semiparametric approach becomes unstable and computationally demanding due to the need to estimate both a $p\times p$ and a $p \times \sum_{r = 1}^R \left(1 + d_r + {d_r \choose 2}\right)$ matrix.

In contrast, IPI uses a much smaller, structured subspace,
\[
\mathcal{A}_{\IPI}(O,M)
=\Big\{\sum_{r=1}^R \lambda_r\Big(\tfrac{\mathbf{1}\{M=\mathbf{1}_d\}}{\mathbb{P}[M = \mathbf{1}_d]}-\tfrac{\mathbf{1}\{M=m_r\}}{\mathbb{P}(M = m_r)}\Big)\,\nabla_\theta \ell\!\big(f(O\!\circ\! m_r);\theta\big):~\lambda_r\in\mathbb{R}\Big\},
\]
the linear span of $R$ pattern-specific, mean-zero functions. This reduces tuning to $\{\lambda_r\}_{r=1}^R$, yielding a more stable and tractable procedure in high-dimensional or complex settings. (A preconditioner $C_1$ can be included analogously; we set $C_1=I$ for simplicity in our implementation.)

\subsection{Experimental comparisons under the factor model}\label{app:aipw_expts}

We present results for the factor model setting comparing IPI with the AIPW baseline.  For this demonstration, we fix the ratio of partially observed to fully observed to be $\tilde{N}/n = 10$. As before, the number of features is $d = 20$ and the number of latent factors is $q = 2$. Results are reported over 500 trials.

\begin{figure}[H]
    \centering
    \includegraphics[width=0.8\linewidth]{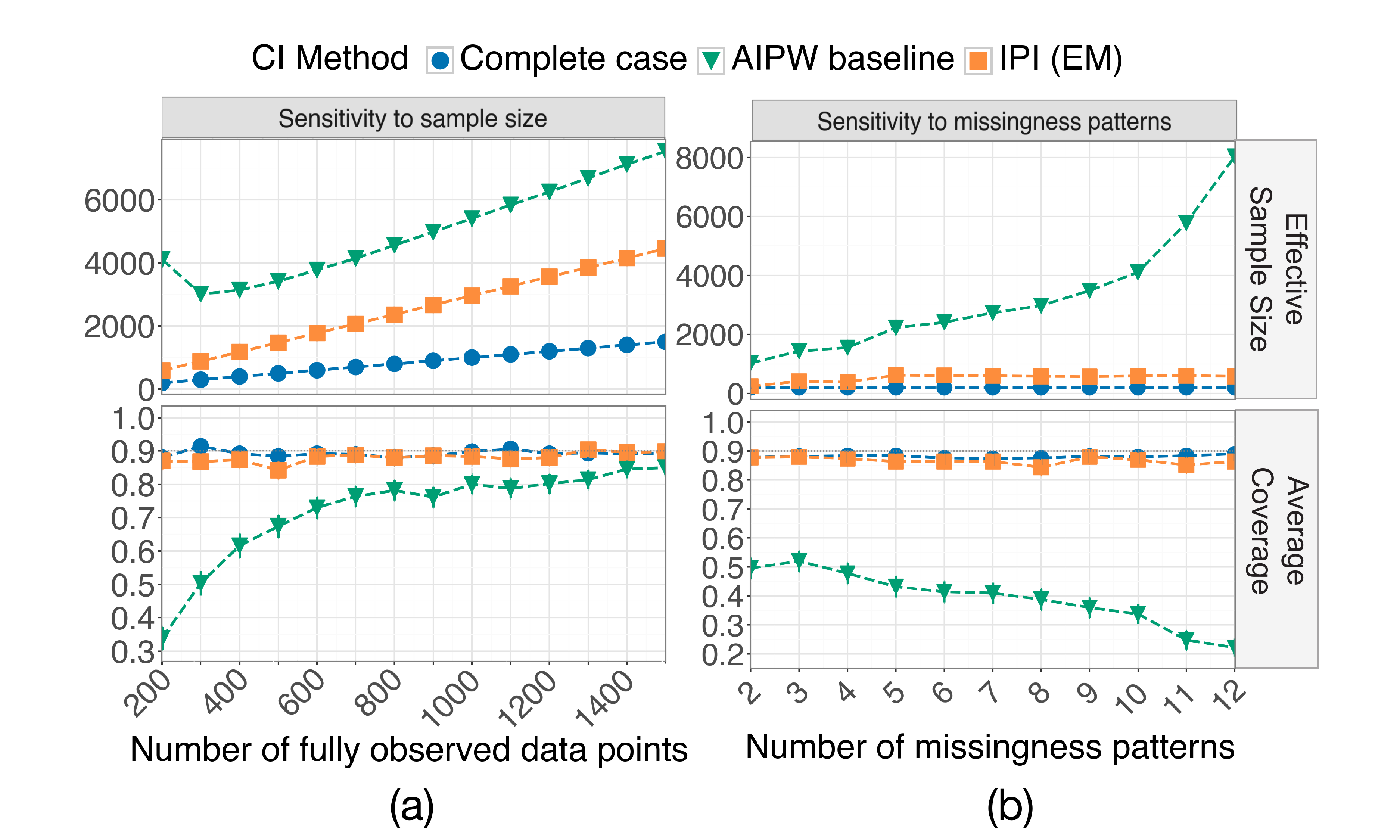}
    \caption{Effective sample size and average coverage for complete-case, AIPW baseline, and IPI tuned methods. (a) varies the number of samples keeping the number of missing patterns to 10 and (b) varies the number of missing patterns keeping the number of fully observed samples to 200.}
    \label{fig:aipw_expts}
\end{figure}

First, we fix the number of missingness patterns to be 10 and increase the number of fully observed samples from 200 to 1000. We see that in Figure \ref{fig:aipw_expts}(a) that the complete-case and IPI methods maintain nominal coverage while the AIPW method requires a large sample size to achieve 90\% average coverage.

Secondly, we fix the number of fully observed samples to be 200 and vary the number of missingness patterns from 2 to 12. We see in Figure \ref{fig:aipw_expts}(b) that as the number of missingness patterns increases, there is a downward trend in the AIPW method  coverage, whereas IPI still achieves relatively close to nominal coverage.

\section{IPI asymptotic normality and confidence interval validity under MCAR}\label{app:IPI}
In this section, we provide the detailed proof for the asymptotic normality of the IPI estimator and the asymptotic validity of the IPI confidence interval. 

We begin by introducing the regularity conditions for our results. Assumption \ref{assump:regularity} enforces common regularity assumptions on our losses. These assumptions permit a variety of losses, including common general linear model log-likelihoods such as squared loss and cross-entropy loss. 

A key distinction from the PPI proofs is the need for conditional central limit theorem (CLT) assumptions (Assumption \ref{assump:cond_lindeberg_mloss}). This arises because masking fully observed data using missing patterns from partially observed data introduces dependence between the fully observed and partially observed sample means. However, conditional on the observed missing patterns, the sample means become conditionally independent. Assumption \ref{assump:cond_lindeberg_mloss} ensures that, under this conditioning, the samples meet standard conditions for asymptotic normality.

\begin{assumption}[Loss regularity conditions]\label{assump:regularity} We state common regularity conditions for the loss and imputation algorithm $f$. 
    \begin{enumerate}
        \item $\ell(X; \theta)$ is a convex function of $\theta$ $P_X$-a.e.
        \item The Hessian $\nabla^2\ell$ exists $P_{D, M}$-a.e. and is locally Lipschitz in the sense that for all $\theta$ in a neighborhood $\mathcal{U}$ of $\theta_\star$, for each missing pattern $r \in [R]$,
        \begin{equation}
            \|\nabla^2 \ell (f(O \circ m_r); \theta) - \nabla^2 \ell(f(O \circ m_r); \theta_\star)\|_{op} \leq B^{(r)} (O\circ m_r) \|\theta - \theta_\star\|_2,
        \end{equation}
        $P_D$-a.e., where $B^{(r)}(O \circ m_r)$ is square integrable.
        \item The population Hessian is invertible: $H_{\theta_\star} := \mathbb{E}[\nabla^2 \ell(X; \theta_\star)] \succ 0$.
    \end{enumerate}
\end{assumption}

\begin{assumption}[Conditional Lindeberg conditions - from \cite{Bulinski2017}]\label{assump:cond_lindeberg_mloss} Here, we provide the regularity conditions for asymptotic normality of the estimator that allow us to construct asymptotic level $1-\alpha$ confidence intervals.

Given the $\sigma$-algebra $\mathcal{F}^M_N := \sigma \left(\{M_i\}_{i =1}^N\right)$, i.e., the $\sigma$-algebra generated from the set of observed missing patterns, define
\begin{align*}
    W_{N, i} &:= \begin{cases}
        \frac{\lambda_r}{R\tilde{N}_r}\nabla \ell(f(O_i); \theta), &\text{\em for } M_i = m_r, r \in [R], \\
        \frac{1}{n} \left(\nabla \ell (X_i; \theta) - \frac{1}{R}\sum_{r = 1}^R \lambda_r \bm{1}\{\tilde{N}_r > 0\}\nabla \ell( f(O_{i}\circ m_r); \theta)\right), & \text{\em for } M_i = \bm{1}_d. 
     \end{cases}
\end{align*}

Note that the gradient of the IPI loss estimator is exactly the sum of the above terms, i.e.:
\begin{align*}
    \nabla L_{\IPI}(\theta; \lambda) = \sum_{i = 1}^{N} W_{N, i}.
\end{align*}

The corresponding conditional means are then defined as
\begin{align*}
  \mu_{N, i, \mathcal{F}^M} &:= \mathbb{E}[W_{N, i}|\mathcal{F}^M] = \mathbb{E}[W_{N, i} | \mathcal{F}_N^M],\\
\mu_{N, \mathcal{F}^M} &:= \mathbb{E}[\nabla L_{\IPI}(\theta; \lambda)|\mathcal{F}^M] = \sum_{i = 1}^{N} \mu_{N, i, \mathcal{F}^M}.
\end{align*}

Moreover, the conditional variances are
\begin{align*}
\Sigma_{N, i, \mathcal{F}^M} &:=\Var(W_{N, i}|\mathcal{F}^M) < \infty \,\, \text{\em  a.e.}\\
S_{N, \mathcal{F}^M} &:= \Var(\nabla L_{\IPI}(\theta; \lambda)|\mathcal{F}^M) = \sum_{i = 1}^{N} \Sigma_{N, i, \mathcal{F}^M} \succ 0,
\end{align*}
where the last equality follows from conditional independence of the terms $W_{N, i}$ given $\mathcal{F}^M$.

The \textbf{conditional Lindeberg condition} is given by
\begin{equation}
    \sum_{i = 1}^{N} \mathbb{E}\left[\left\|{S^{-1/2}_{N, \mathcal{F}^M}}(W_{N, i} - \mu_{N, i, \mathcal{F}^M})\right\|^2\mathbf{1}\{\|S^{-1/2}_{N, \mathcal{F}^M}(W_{N, i} - \mu_{N, i, \mathcal{F}^M})\| > \varepsilon \}\bigg| \mathcal{F}^M\right] \overset{p}{\rightarrow} 0
\end{equation}
for any fixed $\varepsilon > 0$ as $N \rightarrow \infty$.
\end{assumption}

% \ejc{You seem to use a.e. and a.s. interchangeably. Please pick one.}\smz{addressed!}
\subsection{Proof of Theorem \ref{thm:ipi_clt_mcar}} \label{app:proof_ipi_clt_mcar}
Here, we provide a proof of Theorem \ref{thm:ipi_clt_mcar}, which follows conventional arguments for one-step estimators \cite{Vaart_1998}.

First, we fix $\lambda$. Given the above regularity conditions and MCAR assumption, 
\begin{equation}
    \sqrt{n} (\hat{\theta}_{\IPI, {\lambda}} - \theta_\star) \overset{d}{\rightarrow} \mathcal{N}(0, \Sigma),
\end{equation}
where we recall that  
\begin{align*}
    \hat{\theta}_{\IPI, {\lambda}} = \hat{\theta}_n - \hat{H}_n^{-1} \nabla L_\IPI(\hat{\theta}; \lambda), \quad \Sigma &= H_{\theta_\star}^{-1}V_{\theta_\star, \lambda} H_{\theta_\star}^{-1},
\end{align*}
and
\begin{multline*}
    V_{\theta_\star, \lambda}
    := \operatorname{Var} \left(\nabla \ell(X; {\theta_\star}) - \frac{1}{R}\sum_{r=1}^R \lambda_r \nabla \ell(f(O\circ m_r); {\theta_\star})\Bigg|M = \bm{1}_d\right)\\
    + \sum_{r=1}^R \left(\frac{\lambda_r}{R}\right)^2 \frac{p_0}{p_r} \Var\left(\nabla \ell(f(O); {\theta_\star})\mid M = m_r\right).
\end{multline*}

\begin{proof}
The proof follows from standard arguments for a one-step estimator \cite{Vaart_1998}. We first conduct the following decomposition:

\begin{align}              
\hat{\theta}_{\IPI, \lambda} - \theta_\star
&= \hat{\theta}_n - \theta_\star - \hat{H}_n^{-1} \nabla L_{\IPI}(\hat{\theta}_n; \lambda) \nonumber\\
&= (\hat{\theta}_n - \theta_\star)\\
&\hspace{3mm} - \hat{H}_n^{-1} \left(\nabla L_{\IPI}(\hat{\theta}_n; \lambda) - \nabla L_{\IPI}(\theta_\star; \lambda)\right)\\
&\hspace{3mm} - \hat{H}_n^{-1} \nabla L_{\IPI}(\theta_\star; \lambda).
\end{align}

\begin{lemma}
    $\nabla L_{\IPI}(\hat{\theta}_n; \lambda) - \nabla L_{\IPI}(\theta_\star; \lambda) = H_{\theta_\star}(\hat{\theta}_n - \theta) + o_p(1/\sqrt{n})$.
\end{lemma}
\begin{proof}
    This essentially follows from the local Lipschitz Assumption \ref{assump:regularity}. In particular, 
    \begin{multline*}
    \nabla L_{\IPI}(\hat{\theta}_n; \lambda) - \nabla L_{\IPI}(\theta_\star; \lambda)
    = \mathbb{P}_n \left[\nabla \ell(X; \hat{\theta}_n) - \nabla \ell(X; {\theta}_*)\right]\\
    +\sum_{r =1}^R \frac{\lambda_r}{R} \bigg(
    \mathbb{P}_{\tilde{N}_r} \left[\nabla \ell(f(O); \hat{\theta}_n) - \nabla \ell(f(O); {\theta}_\star)\right] - \mathbb{P}_n\left[\nabla \ell(f(O \circ m_r); \hat{\theta}_n) - \nabla \ell(f(O \circ m_r); {\theta}_*)\right]
    \bigg).
\end{multline*}
For $\hat{\theta}_n \in \mathcal{U}$, we have for each $r= 0, ..., R$,
\begin{multline*}
    \nabla \ell(f(O \circ m_r); \hat{\theta}_n) - \nabla \ell (f(O \circ m_r); \theta_\star) 
    = \nabla^2\ell(f(O \circ m_r); \theta_\star)(\hat{\theta}_n - \theta_\star) \\
    + \text{ Remainder}^{(r)}(O\circ m_r) (\hat{\theta}_n - \theta_\star),
\end{multline*}
where by local Lipschitzness of the Hessian (Assumption \ref{assump:regularity}), we have
\begin{align*}
    \left\|\text{ Remainder}^{(r)}\right\|_{op} \leq \frac{B^{(r)}(O \circ m_r)}{2} \|\hat{\theta}_n - \theta_\star\|_2.
\end{align*}

Therefore, 
\begin{align*}
&\nabla L_{\IPI}(\hat{\theta}_n; \lambda) - \nabla L_{\IPI}(\theta_\star; \lambda) = \left(\nabla^2 L_{\IPI}(\theta_\star;\lambda)  + O_p(\|\hat{\theta}_n - \theta_\star\|_2)\right)\left(\hat{\theta}_n - \theta_\star\right).
\end{align*}

By $\sqrt{n}$-consistency of $\hat{\theta}_n$, we have $O_p(\|\hat{\theta}_n - \theta_\star\|_2) = O_p(1/\sqrt{n}) = o_p(1)$ and $\mathbb{P} [\hat{\theta}_n \in \mathcal{U}] \rightarrow 1$. Moreover, by standard law of large numbers and MCAR, $\nabla^2 L_\IPI(\theta_\star; \lambda) = H_{\theta_*} + o_p(1)$.

Therefore, putting it together, 
\begin{align*}
    \nabla L_{\IPI}(\hat{\theta}_n; \lambda) - \nabla L_{\IPI}(\theta_\star; \lambda) = H_{\theta_\star} \left(\hat{\theta}_n - \theta_\star\right) + o_p\left(\frac{1}{\sqrt{n}}\right).
\end{align*}
\end{proof}

\begin{lemma}\label{lemma:grad_lindeberg_clt}
Assumption  \ref{assump:cond_lindeberg_mloss} and MCAR give 
\begin{align*}
    \sqrt{n} \left(V_{\theta_\star, \lambda}\right)^{-1/2}\nabla L_{\IPI}(\theta_\star; \lambda) \overset{d}{\rightarrow} \mathcal{N}(0, I_d).
\end{align*}
\end{lemma}
\begin{proof}
By the conditional Lindeberg CLT (\cite{Bulinski2017}) (Assumptions \ref{assump:cond_lindeberg_mloss} and MCAR hold),  
\[
    \Var(\nabla L_\IPI(\theta_\star; \lambda))^{-1/2} \nabla L_\IPI(\theta_\star; \lambda) \overset{d}{\rightarrow} \mathcal{N}(0, I_d).
\]

Note that given $\mathcal{F}^M,$ the sample means for each missing pattern are conditionally independent. Then,
\begin{multline*}
    n\Var\left(\nabla L_{\IPI}(\theta_\star; \lambda)|\mathcal{F}^M\right) 
    = n\Bigg(\sum_{r = 1}^R \frac{\lambda_r^2}{R^2} \cdot \frac{1}{\tilde{N}_r} \Var \left(\nabla \ell(f(O);\theta_\star)|M =m_r\right) \\  + \frac{1}{n}\Var\left(\nabla \ell(f(O);\theta_\star) - \sum_{r =1}^R \frac{\lambda_r}{R}\nabla \ell(f(O \circ m_r);\theta_\star) \bigg| M = \bm{1}_d \right)\Bigg).
\end{multline*}
Taking the total number $N$ of samples to $\infty$ yields 
\begin{align*}
    n\Var\left(\nabla L_{\IPI}(\theta_\star; \lambda)|\mathcal{F}^M\right)  = V_{\theta_\star, \lambda} +o_p(1).
\end{align*}
Thus, by Slutsky's theorem, 
\[
     \sqrt{n} \left(V_{\theta_\star, \lambda}\right)^{-1/2}  \Var(\nabla L_\IPI(\theta_\star; \lambda))^{1/2}  \Var(\nabla L_\IPI(\theta_\star; \lambda))^{-1/2}\nabla L_{\IPI}(\theta_\star; \lambda) \overset{d}{\rightarrow} \mathcal{N}(0, I_d).
\]
\end{proof}

Putting all the components together yields
\begin{align*}
     \hat{\theta}_{\IPI, \lambda} - \theta_\star
    &= (\hat{\theta}_n - \theta_\star) - \hat{H}_n^{-1} \left(\nabla L_{\IPI}(\hat{\theta}_n; \lambda) - \nabla L_{\IPI}(\theta_\star; \lambda)\right) - \hat{H}_n^{-1} \nabla L_{\IPI}(\theta_\star; \lambda)\\
    &= \hat{\theta}_n - \theta_\star - (\hat{\theta}_n - \theta_\star) + o_p(1/\sqrt{n})  - H_{\theta_*}^{-1}\nabla L_\IPI(\theta_\star; \lambda) + o_p(1/\sqrt{n}) \\
    &= -H_{\theta_*}^{-1}\nabla L_\IPI(\theta_\star; \lambda) + o_p(1/\sqrt{n}).
\end{align*}

Thus,
\begin{align*}
    \sqrt{n}(\hat{\theta}_{\IPI, {\lambda}} - \theta_\star) 
    &=- \sqrt{n}H_{\theta_\star}^{-1}\Var\left(\nabla L_{\IPI}(\theta_\star; \lambda)|\mathcal{F}^M\right)^{1/2}\Var\left(\nabla L_{\IPI}(\theta_\star; \lambda)|\mathcal{F}^M\right)^{-1/2} \nabla L_\IPI (\theta_\star; \lambda)\\
    &\overset{d}{\rightarrow} \mathcal{N}(0, \Sigma),
\end{align*}
where the last line follows from Slutsky's theorem.

For a random sequence $\hat{\lambda}$ converging to $\lambda$ in probability, we still achieve the above asymptotic behavior via adding and subtracting one more term:
\begin{align*}              
\hat{\theta}_{\IPI, \hat{\lambda}} - \theta_\star
&= \hat{\theta}_n - \theta_\star - \hat{H}_n^{-1} \nabla L_{\IPI}(\hat{\theta}_n; \hat{\lambda}) \nonumber\\
&= (\hat{\theta}_n - \theta_\star)\\
&\hspace{3mm} - \hat{H}_n^{-1} \left(\nabla L_{\IPI}(\hat{\theta}_n;  \hat{\lambda}) - \nabla L_{\IPI}(\theta_\star;  \hat{\lambda})\right)\\
&\hspace{3mm} - \hat{H}_n^{-1} \left(\nabla L_{\IPI}(\theta_\star;  \hat{\lambda}) -\nabla L_{\IPI}(\theta_\star;  {\lambda})  \right)\\
&\hspace{3mm} -\hat{H}_n^{-1}\nabla L_{\IPI}(\theta_\star;  {\lambda}).
\end{align*}

Since $L_\IPI$ is linear in $\lambda$, by MCAR, $\nabla L_\IPI (\hat{\theta}_n; \hat{\lambda}) - \nabla L_\IPI (\theta_\star; \hat{\lambda}) = (H_{\theta_\star} + o_p(1))(\hat{\theta}_n - \theta_\star)$. 
\begin{align*}
    \nabla L_{\IPI}(\theta_\star;  \hat{\lambda}) -\nabla L_{\IPI}(\theta_\star;  {\lambda})
    &= \frac{1}{R}\sum_{r =1}^R (\hat{\lambda}_r - \lambda_r) \left(\mathbb{P}_{\tilde{N}_r} [\nabla \ell(f(O); \theta_\star)] - \mathbb{P}_n [\nabla \ell(f(O \circ m_r); \theta_\star)]\right)
\end{align*}
is $o_p(1/\sqrt{n})$ if for each $r \in [R]$, $\mathbb{P}_{\tilde{N}_r} [\nabla \ell(f(O); \theta_\star)] - \mathbb{P}_n [\nabla \ell(f(O \circ m_r); \theta_\star)] = O_p(1/\sqrt{n})$, which is guaranteed by MCAR and conditional Lindeberg CLT.
\end{proof}

The proof of Corollary \ref{cor:ipi_clt_mcar} follows immediately from Theorem \ref{thm:ipi_clt_mcar}.
\begin{proof}(of Corollary \ref{cor:ipi_clt_mcar}).
    \begin{align*}
        \lim_{n \rightarrow \infty} \mathbb{P} (\theta_\star \in \mathcal{C}_{\IPI, \alpha}) = \lim_{n \rightarrow \infty} \mathbb{P}\left(n(\hat{\theta}_{\IPI, \hat{\lambda}} - \theta_\star)^\top \Sigma^{-1} (\hat{\theta}_{\IPI, \hat{\lambda}} - \theta_\star) \leq \chi_{p, 1-\alpha}^2\right) = 1-\alpha,
    \end{align*}
    by definition of convergence in distribution.
\end{proof}

\section{CIPI asymptotic normality and confidence interval validity}\label{app:CIPI}

In this section, we extend our results to cross-fitting. To accommodate the more general setting, we slightly refine the assumptions underlying IPI validity. We now require the regularity conditions to hold for every cross-fitted imputation function $f^{(k)}$ and the long-run average imputation $\bar{f}$ (Assumption \ref{assump:cipi_regularity}). The conditional CLT condition holds also for $f^{(k)}$ and $\bar{f}$ (Assumption \ref{assump:cond_lindeberg_cipi}).

\begin{assumption}[Loss regularity conditions for CIPI]\label{assump:cipi_regularity} We state the regularity conditions for the loss under cross-fitting.
    \begin{enumerate}
        \item $\ell(X; \theta)$ is a convex function of $\theta$ $P_X$-a.e.
        
        \item The Hessian $\nabla^2 \ell$ is locally Lipschitz  in the sense of Assumption \ref{assump:regularity} for all $r \in [R]$ and any imputation function $f^{(k)}$ and average $\bar{f}$. In particular, we have square-integrable Lipschitz bounds for each of $f^{(k)}$ and $\bar{f}$.

        \item $H_{\theta_\star} := \mathbb{E} \left[\nabla^2 \ell(X; \theta_\star)\right] \succ 0$. 
        
        For each $r \in [R]$ and $\tilde{r} \in \{0, r\}$, $\mathbb{E}\left[\|\nabla^2 \ell(f^{(1)}(O \circ m_r);\theta_\star)\||f^{(1)}, M= m_{\tilde{r}}\right] = O_p(1)$.
    \end{enumerate}
\end{assumption}
\begin{assumption}[Conditional Lindeberg conditions for CIPI]\label{assump:cond_lindeberg_cipi} Here, we provide the regularity conditions for asymptotic normality of the estimator that allow us to construct asymptotic level $1-\alpha$ confidence intervals when using cross-fitting.

\begin{itemize}
    \item The conditional Lindeberg condition (Assumption \ref{assump:cond_lindeberg_mloss}) holds for $\bar{f}$.
    \item For each of $f^{(k)}$, $k \in [K]$, the corresponding conditional Lindeberg condition holds. In particular, let $\mathcal{F}_{N}^{M,k} := \sigma (\{M_i\}_{i =1}^N, f^{(k)})$ be the $\sigma$-algebra generated from the set of observed missing patterns and cross-fitted imputer $f^{(k)}$. Define for each $r \in [R]$, $k \in [K]$ with the $k$th fold denoted $\mathcal{I}_k$,
    \begin{align*}
    W_{N, i}^{r,k} &:= \begin{cases}
        \frac{1}{\tilde{N}_{r,k}}\bm{1}\{\tilde{N}_{r,k} > 0\}\nabla \ell(f^{(k)}(O_i); \theta), &\text{\em for } M_i = m_r, i \in \mathcal{I}_k \\
        - \frac{1}{n_k} \bm{1}\{\tilde{N}_{r,k} > 0\}\nabla \ell( f^{(k)}(O_{i}\circ m_r); \theta), & \text{\em for } M_i = \bm{1}_d, i \in \mathcal{I}_k\\
        0 & \text{\em otherwise}
     \end{cases}
    \end{align*}
    Note that given $\mathcal{F}_M^{M, k}$, $W_{N, i}^{r,k}$ are conditionally independent.
    
    Following the presentation in Assumption \ref{assump:cond_lindeberg_mloss}, we note 
    \[
        \mathbb{P}_{\tilde{N}_{r,k}}[\nabla \ell(f^{(k)}(O); \theta)] - \mathbb{P}_{n_k}[\nabla \ell(f^{(k)}(O \circ m_r); \theta)] = \sum_{i =1}^N W_{N,i}^{r,k},
    \]
    and denote the conditional means and variances as:
    \begin{align*}
        \mu^r_{N, i, \mathcal{F}^{M,k}} &:= \mathbb{E}[W_{N, i}^{r,k}|\mathcal{F}^{M,k}] = \mathbb{E}[W_{N, i}^{r,k} | \mathcal{F}_N^{M,k}], &\mu^r_{N, \mathcal{F}^{M,k}} := \sum_{i = 1}^{N} \mu^r_{N, i, \mathcal{F}^{M,k}}, \\
        \Sigma^r_{N, i, \mathcal{F}^{M,k}} &:=\Var(W^{r,k}_{N, i}|\mathcal{F}^{M,k}) < \infty \,\, \text{\em  a.e.},
        &S_{N, r, \mathcal{F}^{M,k}} := \sum_{i = 1}^{N} \Sigma^r_{N, i, \mathcal{F}^{M,k}} \succ 0.
    \end{align*}
    Then, the conditional Lindeberg condition for $r$ and $k$ is given by
    \begin{equation}
    \sum_{i = 1}^{N} \mathbb{E}\left[\left\|{S^{-1/2}_{N, r, \mathcal{F}^{M,k}}}(W^{r,k}_{N, i} - \mu^r_{N, i, \mathcal{F}^{M, k}})\right\|^2\mathbf{1}\{\|S^{-1/2}_{N,r,  \mathcal{F}^{M,k}}(W_{N, i}^{r,k} - \mu^r_{N, i, \mathcal{F}^{M,k}})\| > \varepsilon \}\bigg| \mathcal{F}^{M,k}\right] \overset{p}{\rightarrow} 0.
    \end{equation}
\end{itemize}
\end{assumption}

\begin{theorem}[CIPI CLT under MCAR]\label{thm:cipi_clt_mcar} Assume MCAR, regularity (Assumption \ref{assump:cipi_regularity}), stability (Assumption \ref{assump:stability assumptions}), and the conditional Lindeberg condition (Assumption \ref{assump:cond_lindeberg_cipi}). Then, if $\hat{\lambda} = \lambda + o_p(1)$, 
\begin{align*}
    \sqrt{n} \left(\hat{\theta}_{\CIPI, \hat{\lambda}} - \theta_\star\right) \overset{d}{\rightarrow} \mathcal{N}(0, \bar{\Sigma}),
\end{align*}
where
\begin{equation}
    \bar{\Sigma} := H_{\theta_\star}^{-1} \bar{V}_{\theta_\star, \lambda} H_{\theta_\star}^{-1},
\end{equation}
and
\begin{multline}\label{eq:cipi_asymp_var}
    \bar{V}_{\theta_\star, \lambda}
    := \Var \left(\nabla \ell(X; {\theta_\star}) - \frac{1}{R}\sum_{r=1}^R \lambda_r \nabla \ell(\bar{f}(O\circ m_r); {\theta_\star})\bigg| M = \bm{1}_d\right)\\
    + \sum_{r=1}^R \left(\frac{\lambda_r}{R}\right)^2 \frac{p_0}{p_r} \Var\left(\nabla \ell(\bar{f}(O); {\theta_\star})\mid M = m_r\right).
\end{multline}
\end{theorem}
\begin{proof}
Like in the IPI setting, we consider the following quantity:
\begin{align}\label{eq:cipi-decomp1}
    \hat{\theta}_{\CIPI, \hat{\lambda}} - \theta_\star
    &= \hat{\theta}_n - \theta_\star - \hat{H}_n^{-1} \nabla L_\CIPI(\hat{\theta}_n; \hat{\lambda})\\
    &= -\hat{H}_n^{-1}\nabla L_\IPI^{\bar{f}}(\theta_\star; \lambda) + \underbrace{\hat{\theta}_n - \theta_\star - \hat{H}_n^{-1} \left(\nabla L_\CIPI(\hat{\theta}_n; \hat{\lambda})- \nabla L_\IPI^{\bar{f}}(\theta_\star; \lambda)\right)}_{\text{we will show: }o_p(1/\sqrt{n})},\nonumber
\end{align}
where $L_{\IPI}^{\bar{f}}$ is the IPI loss with imputation algorithm $\bar{f}$. 

For convenience of presentation, assume $n$ and $\tilde{N}_r$ are divisible by $K$. Let $\mathbb{P}_{n_k}g$ denote the empirical mean of a function $g$ over fully observed samples in fold $k$, and $\mathbb{P}_{\tilde{N}_{r,k}}g$ denote the empirical mean of function $g$ over the samples in fold $k$ with missing pattern $r$.

With this notation in hand, consider: 

\begin{align}\label{eq:cipi-decomp2}
    \nabla L_\CIPI (\hat{\theta}_n; \hat{\lambda}) - \nabla L_\IPI^{\bar{f}}({\theta}_\star; {\lambda}) 
    &= \nabla L_\CIPI(\hat{\theta}_n; \hat{\lambda}) - \nabla L_\CIPI(\theta_\star; \hat{\lambda})\\
    &\hspace{3mm} + \nabla L_\CIPI(\theta_\star; \hat{\lambda}) - \nabla L_\CIPI(\theta_\star; \lambda)\nonumber\\
    &\hspace{3mm} + \nabla L_\CIPI(\theta_\star; \lambda) - \nabla L_\IPI^{\bar{f}}({\theta}_\star; {\lambda}).\nonumber
\end{align}

For the first term, since $\hat{\theta}_n$ is $\sqrt{n}$-consistent for $\theta_\star$, by local Lipschitzness of the Hessian and MCAR, for each $k \in [K]$,
\begin{align*}
    \nabla L_\IPI^{(k)}(\hat{\theta}_n; \hat{\lambda}) - \nabla L_\IPI^{(k)}(\theta_\star; \hat{\lambda})
    &= \left(\mathbb{P}_{n_k} \left[\nabla^2 \ell(X; \theta_\star) - \frac{1}{R} \sum_{r =1}^R\hat{\lambda}_r \nabla^2 \ell(f^{(k)}(O \circ m_r); \theta_\star)\right]  + o_p(1)\right) (\hat{\theta}_n - \theta_\star)\\
    &\hspace{5mm} + \frac{1}{R} \sum_{r =1}^R \hat{\lambda}_r \left(\mathbb{P}_{\tilde{N}_{r,k}}\left[\nabla^2 \ell( f^{(k)}(O); \theta_\star)\right] + o_p(1)\right)(\hat{\theta}_n - \theta_\star)\\
    &= H_{\theta_\star} (\hat{\theta}_n - \theta_\star) + o_p(1/\sqrt{n}) .
\end{align*}

Thus, 
\[
 \nabla L_\CIPI(\hat{\theta}_n; \hat{\lambda}) - \nabla  L_\CIPI({\theta}_\star; \hat{\lambda}) = H_{\theta_\star}(\hat{\theta}_n - \theta_\star) + o_p(1/\sqrt{n}).
\]

Similar to the proof for Theorem \ref{thm:ipi_clt_mcar} in Section \ref{app:proof_ipi_clt_mcar}, we show that the second term is $o_p(1/\sqrt{n})$ by MCAR. In particular, for each $k \in [K]$,

\begin{align*}
    \nabla L_\IPI^{(k)}(\theta_\star; \hat{\lambda}) - \nabla L_\IPI^{(k)} (\theta_\star; \lambda)
    &= \frac{1}{R}\sum_{r =1}^R o_p(1) \underbrace{\left(\mathbb{P}_{\tilde{N}_{r,k}}\left[\nabla \ell(f^{(k)}(O); \theta_\star)\right] - \mathbb{P}_{n_k} \left[\nabla \ell(f^{(k)}(O \circ m_r); \theta_\star)\right]\right)}_{O_p(1/\sqrt{n})}\\ &= o_p(1/\sqrt{n}).
\end{align*}

Lastly, to show the last term is $o_p(1/\sqrt{n})$, we adapt proofs from \citet{zrnic_cppi} and \cite{bayle2020cross}. In particular, 
\begin{align*}
    &\frac{1}{K}\sum_{k =1}^K \left( \nabla L_\IPI^{(k)}(\theta_\star; {\lambda}) - \nabla L_\IPI^{\bar{f}}(\theta_\star; {\lambda})\right)\\ 
    &= \frac{1}{R}\sum_{r =1}^R {\lambda}_r \bigg[\frac{1}{K}\sum_{k =1}^K\left(\mathbb{P}_{\tilde{N}_{r,k}}\left[\nabla \ell(f^{(k)}(O); \theta_\star) - \nabla \ell(\bar{f}(O); \theta_\star)\right] \right)\\
    &\hspace{20mm} -  \frac{1}{K}\sum_{k =1}^K\left(\mathbb{P}_{n_k}\left[\nabla \ell(f^{(k)}(O \circ m_r); \theta_\star) - \nabla \ell(\bar{f}(O \circ m_r); \theta_\star)\right] \right)\bigg].
\end{align*}
We will now show that for each $r$, the centered term within the bracket is $o_p(1/\sqrt{n})$.

\begin{lemma}\label{lemm:cross-fit-partial}
    Under MCAR and Assumption \ref{assump:stability assumptions}, for each $r \in [R]$, 
    \begin{align*}
        \frac{1}{K}\sum_{k =1}^K \bigg(&\mathbb{P}_{\tilde{N}_{r,k}} \left[\nabla \ell\left(f^{(k)}(O); \theta_\star\right) - \nabla \ell(\bar{f}(O); \theta_\star)\right] \\
        &\hspace{3mm} - \left(\mathbb{E}\left[\nabla \ell (f^{(k)}(O); \theta_\star) -\nabla \ell(\bar{f}(O); \theta_\star) | f^{(k)}, M = m_r\right]\right) \bigg) = o_p(1/\sqrt{n}). 
    \end{align*}
\end{lemma}
\begin{proof}
    For ease of notation, let 
    \begin{align*}
        \tilde{F}_{r,k} := \sqrt{\tilde{N}_r}\bigg(&\mathbb{P}_{\tilde{N}_{r,k}} \left[\nabla \ell\left(f^{(k)}(O); \theta_\star\right) - \nabla \ell(\bar{f}(O); \theta_\star)\right] \\
        &\hspace{3mm} - \left(\mathbb{E}\left[\nabla \ell (f^{(k)}(O); \theta_\star) -\nabla \ell(\bar{f}(O); \theta_\star) \big| f^{(k)}, M = m_r\right]\right) \bigg). 
    \end{align*}

    Then, since $\sqrt{{n}/{\tilde{N}_r}} = O_p(1)$, we aim to show: $\frac{1}{K}\sum_{k =1}^K \tilde{F}_{r,k} = o_p(1)$. Like in \citet{zrnic_cppi}, we use the following fact proved in \citet{bayle2020cross} Lemma 2.
    \begin{fact}
    Let $\psi(x) = \min(1,x).$ Then, for a sequence of random variables $X_n$, $X_n = o_p(1)$ if and only if $\mathbb{E}\left[\psi(|X_n|)\right] = o(1)$.      
    \end{fact}

    Thus, we aim to show $\mathbb{E} \left[\psi \left(\left|\frac{1}{K}\sum_{k =1}^K \tilde{F}_{r,k}\right|\right)\right] \rightarrow 0$. First, we leverage the fact that (i) $\psi$ is monotone increasing; (ii) $\psi(\sum_k |x_k|) \leq \sum_k \psi(|x_k|)$; and (iii) $\psi$ is concave.

    \begin{align*}  \mathbb{E}\left[\psi\left(\left|\frac{1}{K}\sum_{k=1}^K \tilde{F}_{r,k}\right|\right)\right] \leq \sum_{k =1}^K \mathbb{E}\left[\psi\left(\frac{1}{K}|\tilde{F}_{r,k}|\right)\right]
        \leq  \sum_{k =1}^K \mathbb{E}\left[\psi\left(\frac{1}{K}\mathbb{E}\left[|\tilde{F}_{r,k}|\big|f^{(k)}, \mathcal{F}^M\right]\right)\right].
    \end{align*}
    Then, by Cauchy-Schwarz,
    \begin{align*}
    \mathbb{E}\left[\psi\left(\left|\frac{1}{K}\sum_{k=1}^K \tilde{F}_{r,k}\right|\right)\right] 
        &\leq \sum_{k =1}^K \mathbb{E} \left[\psi\left(\frac{1}{K}\sqrt{\mathbb{E}\left[|\tilde{F}_{r,k}|^2\big|f^{(k)}, \mathcal{F}^M\right]}\right)\right]\\
        &= \sum_{k=1}^K \mathbb{E}\left[\psi \left(\frac{1}{K} \sqrt{\Var(\tilde{F}_{r,k}\big|f^{(k)}, \mathcal{F}^M)}\right)\right]\\
        &= \sum_{k=1}^K \mathbb{E} \left[\psi \left(\frac{1}{K} \sqrt{K \Var\left(\nabla \ell(f^{(k)}(O); \theta_\star) - \nabla \ell(\bar{f}(O); \theta_\star)\big| f^{(k)}, M = m_r\right)}\right)\right]\\
        &= \mathbb{E} \left[\min\left(K, \sqrt{K \Var\left(\nabla \ell(f^{(k)}(O); \theta_\star) - \nabla \ell(\bar{f}(O); \theta_\star)\big| f^{(k)}, M = m_r\right)}\right)\right] \rightarrow 0,
    \end{align*}
    where the last line follows from Assumption \ref{assump:stability assumptions}.  
\end{proof}

The behavior for the analogous term on the fully observed dataset is argued in exactly the same way.
\begin{lemma}\label{lemm:cross-fit-fullyobs}
    Under MCAR and Assumption \ref{assump:stability assumptions}, for each $r \in [R]$, 
    \begin{align*}
        \frac{1}{K}\sum_{k =1}^K \bigg(&\mathbb{P}_{n} \left[\nabla \ell\left(f^{(k)}(O \circ m_r); \theta_\star\right) - \nabla \ell(\bar{f}(O \circ m_r); \theta_\star)\right] \\
        &\hspace{3mm} - \left(\mathbb{E}\left[\nabla \ell (f^{(k)}(O \circ m_r); \theta_\star) -\nabla \ell(\bar{f}(O \circ m_r); \theta_\star) | f^{(k)}, M = \bm{1}_d\right]\right) \bigg) = o_p(1/\sqrt{n}). 
    \end{align*}
\end{lemma}
\begin{proof}
    One can repeat the steps for Lemma \ref{lemm:cross-fit-partial} to prove Lemma \ref{lemm:cross-fit-fullyobs}.
\end{proof}

Since MCAR implies
\begin{multline*}
    \frac{1}{K} \sum_{k =1}^K \mathbb{E}\left[\nabla \ell (f^{(k)}(O); \theta_\star) -\nabla \ell(\bar{f}(O); \theta_\star) | f^{(k)}, M = m_r\right]\\
    = \frac{1}{K}\sum_{k=1}^K\mathbb{E}\left[\nabla \ell (f^{(k)}(O \circ m_r); \theta_\star) -\nabla \ell(\bar{f}(O \circ m_r); \theta_\star) | f^{(k)}, M = \bm{1}_d\right],
\end{multline*}
we have
\begin{align*}
    \frac{1}{K}\sum_{k =1}^K \left( \nabla L_\IPI^{(k)}(\theta_\star; {\lambda}) - \nabla L_\IPI^{\bar{f}}(\theta_\star; {\lambda})\right) = \frac{1}{R}\sum_{r =1}^R {\lambda}_r \left( o_p(1/\sqrt{n}) + 0\right)  = o_p(1/\sqrt{n}).
\end{align*}

Thus, putting all the components together, 
\begin{align*}
    \hat{\theta}_{\CIPI, \hat{\lambda}} - \theta_\star = -H_{\theta_\star}^{-1}\nabla L_\IPI^{\bar{f}}(\theta_\star; \lambda)  + o_p(1/\sqrt{n}).
\end{align*}

The asymptotic behavior of the CIPI estimator reduces to that of the IPI estimator at the average imputer $\bar{f}$. Conditional CLT gives our result.
\end{proof}

As in the IPI case, the confidence interval validity in Corollary \ref{cor:cipi_ci} follows immediately from the asymptotic normality of the estimator and definition of convergence in distribution.

\subsection{Bootstrap variance estimates}\label{app:boot_var_est}
Here, we discuss how we estimate $\bar{V}_{\theta_\star, \lambda}$. To do so, we draw inspiration from \citet{zrnic_cppi} and propose a bootstrap variance approach given in Algorithm \ref{alg:bootstrap_var}. In particular, for each bootstrap trial $b \in [B]$, we train an imputation function $f^{(b)}$ on a bootstrapped sample of size $N(K-1)/K$. Then, for each sample, we approximate $\bar{f}$ by the average imputation over all bootstrapped $f^{(b)}$ that the sample was \textit{not} trained on.

\begin{algorithm}[H]
\caption{Bootstrap variance estimate}
\label{alg:bootstrap_var}
\begin{algorithmic}[1]
\Require
    Observations $\{O_i=((X_i,U_i)\circ M_i)\}_{i=1}^{N}$
\Ensure Estimate of $\bar{V}_{\theta_\star, \lambda}$ 

\smallskip
\For{$b$ in $1$, ...., $B$}
    \State Draw $\mathcal I_b\subseteq[N]$ bootstrap sample of size $\lfloor N(K-1)/K\rfloor$
    \State Fit imputer $f^{(b)}$ on $\{O_i^{*b}:i\in\mathcal I_b\}$
\EndFor

\smallskip
\Statex \textbf{// — Predictions for \emph{fully} observed rows —}
\For{$j$ with $M_j = \bm{1}_d$ and $r$ in $1, ..., R$}
    \State $\bar{f}^{(*B)}(O_j \circ m_r) = \text{mean}\left\{f^{(b)}(O_j \circ m_r): b\in [B] \text{ s.t. }j \notin \mathcal{I}_b \right\}$
\EndFor

\smallskip
\Statex \textbf{// — Predictions for \emph{partially} observed rows —}
\For{$i$ with $M_i \neq \bm{1}_d$}
    \State $\bar{f}^{(*B)}(O_i) = \text{mean}\left\{f^{(b)}(O_i) : b\in [B] \text{ s.t. }i \notin \mathcal{I}_b \right\}$
\EndFor

\State \Return
       \begin{align*}
           &\widehat{\Var}_n \left(\nabla \ell(X; \theta_\star) - \frac{1}{R}\sum_{r =1}^R \hat{\lambda}_r \nabla \ell(\bar{f}^{(*B)}(O \circ m_r); \theta_\star)\right) + \sum_{r=1}^R \left(\frac{\hat{\lambda}_r}{R}\right)^2 \frac{n}{\tilde{N}_r}\widehat{\Var}_{\tilde{N}_r}\left(\nabla \ell(\bar{f}^{(*B)}(O); \theta_\star)\right)
       \end{align*}
\end{algorithmic}
\end{algorithm}

\section{Power tuning}\label{app:Tuning}

In this section, we describe how to choose $\lambda$ so that the asymptotic variance of the IPI estimator is minimized.

For clarity, let us focus on the case where we are estimating one entry of $\theta_\star$, say $(\theta_\star)_1$. We can then choose $\lambda$ as 
\begin{align*}
    \lambda := \argmin_{\lambda} \Sigma_{11}, \quad \Sigma = H_{\theta_\star}^{-1} V_{\theta_\star, \lambda} H_{\theta_\star}^{-1}, 
\end{align*}
and recall Eq.~\eqref{eq:ipi_loss_asymp_var}. 

Thus, our objective is:
\begin{align*}
g(\lambda)
&= \sum_{r = 1}^R \left(\frac{\lambda_r}{R}\right)^2 \frac{1}{\tilde{N}_r} \left(H_{\theta_\star}^{-1} {\Var}\left(\nabla \ell(f(O); \theta_\star)| M = m_r\right)H_{\theta_\star}^{-1}\right)_{11}\\ &\hspace{3mm} + \frac{1}{n} \left(H_{\theta_\star}^{-1} \Var\left(\nabla \ell(X; \theta_\star) - \frac{1}{R}\sum_{r= 1}^R \lambda_r \nabla \ell(f(O \circ m_r); \theta_\star) \bigg| M = \bm{1}_d \right)H_{\theta_\star}^{-1} \right)_{11}.
\end{align*}

Note that $g(\lambda)$ is quadratic in $\lambda$.   To make that explicit, introduce
\begin{align*}
A &:=\operatorname{Diag}\left[\frac{1}{\tilde{N}_r} \left(H_{\theta_\star}^{-1} {\Var}\left(\nabla \ell(f(O); \theta_\star)| M = m_r\right)H_{\theta_\star}^{-1}\right)_{11}\right]_{r \in [R]}\\
C_{rr'} &:=\frac{1}{n}
           \left(H_{\theta_\star}^{-1}
           \Cov\left(\nabla\ell(f(O \circ m_r);\theta_\star),
                \nabla\ell(f(O \circ m_{r'});\theta_\star) | M = \bm{1}_d\right)
           H_{\theta_\star}^{-1}\right)_{11},
\\
b_r &:=\frac{1}{n}
        \left(H_{\theta_\star}^{-1}
        \Cov\left(\nabla\ell(X;\theta_\star),
             \nabla\ell(f(O \circ m_r);\theta_\star)\big| M = \bm{1}_d\right)
        H_{\theta_\star}^{-1}\right)_{11}.
\end{align*}
With these notations in hand, note
\[
g(\lambda)
  \;=\;
  \frac{1}{R^2} \lambda^\top \Bigl(A+C\Bigr)\lambda  - 2 \frac{1}{R}b^\top\lambda
  \;+\;\text{const}.
\]

Thus, the optimal power-tuning parameter has a closed-form solution
\begin{align*}
    \lambda_\star
    &= \left(\frac{1}{R}\left(A+C\right)\right)^{-1}b.
\end{align*}

We estimate each of the above expressions with their empirical counterparts to get $\hat{\lambda}_\star.$ In particular, 

\begin{align*}
    \hat{A} &:=\operatorname{Diag}\left[\frac{1}{\tilde{N}_r} \left(\hat{H}_n^{-1} \widehat{\Var}_{\tilde{N}_r}\left(\nabla \ell(f(O); \theta_\star)\right)\hat{H}_n^{-1}\right)_{11}\right]_{r \in [R]},\\
\hat{C}_{rr'} &:=\frac{1}{n} \left(\hat{H}_n^{-1} \widehat{\Cov}_n\left(\nabla\ell(f(O \circ m_r);\theta_\star), \nabla\ell(f(O \circ m_{r'});\theta_\star) \right) \hat{H}_n^{-1}\right)_{11},
\\
\hat{b}_r &:=\frac{1}{n}
        \left(\hat{H}_n^{-1}\widehat{\Cov}_n\left(\nabla\ell(X;\theta_\star), \nabla\ell(f(O \circ m_r);\theta_\star)\right)
        \hat{H}_n^{-1}\right)_{11},
\end{align*}
and $\hat{\lambda}_\star := \left(\frac{1}{R}(\hat{A}+\hat{C})\right)^{-1}\hat{b}$.

\begin{lemma} Under regularity assumptions (Assumption \ref{assump:regularity}), $\hat{\lambda}_\star - \lambda_* = o_p(1).$
\end{lemma}
\begin{proof}
By the local Lipschitz continuity of the Hessian in
Assumption~\ref{assump:regularity}, together with
\(\hat\theta_n \xrightarrow{p} \theta_\star\),
the sample inverse Hessian \(\hat H_n^{-1}\) is consistent for
\(H_{\theta_\star}^{-1}\).
The weak law of large numbers likewise makes every sample variance or
covariance—whether based on the \(n\) fully observed points or the
\(\tilde N_r\) points in pattern \(r\)—consistent for its population
analogue.
Because each entry of \(\hat A\), \(\hat C\), and \(\hat b\) is a continuous
function of these consistent quantities, we have
\(\hat A \xrightarrow{p} A\), \(\hat C \xrightarrow{p} C\), and
\(\hat b \xrightarrow{p} b\).
Finally, if \(A+C\) is nonsingular (invertible in a neighborhood of its
true value), the map \((U,v)\mapsto U^{-1}v\) is continuous; hence, by the
continuous-mapping theorem, \(\hat\lambda_\star \xrightarrow{p} \lambda_\star\).
\end{proof}

\section{IPI and CIPI under the MCAR first moment assumption}\label{app:away-mcar}
This section develops valid inference procedures for IPI and CIPI estimators under a relaxation of the MCAR assumption. Specifically, we adopt the MCAR first moment assumption (outlined in Assumption \ref{assump:MCAR_first_moment}), which allows for more general missingness mechanisms while still ensuring consistent and asymptotically normal estimators. We begin by analyzing the IPI estimator under this assumption, establishing its asymptotic distribution and diagnostic properties. We then extend these results to the CIPI estimator, incorporating fold-wise imputation and model stability conditions.

\subsection{General IPI}
Note that throughout the MCAR discussion, for ease of presentation, we let our consistent estimator of the Hessian $\hat{H}_n$ be the sample Hessian on the fully observed dataset evaluated at $\hat{\theta}_n$. Here, we consider a more general form of estimator.

\begin{definition}[IPI estimator, general]
    $\hat{\theta}_{\IPI, \hat{\lambda}} = \hat{\theta}_n - \hat{H}_n^{-1} \nabla L_\IPI (\hat{\theta}_n; \hat{\lambda})$, where $\hat{H}_n$ is consistent for 
    \begin{multline*} \tilde{J}_{\tilde{\theta}_\star} := \mathbb{E}\left[\nabla^2 \ell(X; \tilde{\theta}_\star) - \frac{1}{R}\sum_{r = 1}^R \lambda_r \nabla^2 \ell(f(O \circ m_r); \tilde{\theta}_\star) \bigg| M = \bm{1}_d\right]\\ + \frac{1}{R}\sum_{r =1}^R \lambda_r \mathbb{E}\left[\nabla^2 \ell(f(O); \tilde{\theta}_\star) \bigg| M = m_r\right].\end{multline*}
    We assume $\tilde{J}_{\tilde{\theta}_\star} \succ 0$. In the following exposition, we will take $ \hat{H}_n = \nabla^2 L_\IPI(\hat{\theta}_n; \hat{\lambda})$.
\end{definition}

\begin{theorem}\label{thm:general_ipi}
    Assume regularity conditions (Assumption \ref{assump:regularity}), the conditional Lindeberg for IPI (Assumption \ref{assump:cond_lindeberg_mloss}), and the MCAR first moment condition with $f$ (Assumption \ref{assump:MCAR_first_moment}). Then, for $\hat{\lambda} = \lambda + o_p(1)$,

    \begin{align*}
        \sqrt{n} (\hat{\theta}_{\IPI, \hat{\lambda}} - \tilde{\theta}_\star) \overset{d}{\rightarrow} \mathcal{N}(0, \Sigma),
    \quad \Sigma = \tilde{J}_{\tilde{\theta}_\star}^{-1} V_{\tilde{\theta}_\star, \lambda} \tilde{J}_{\tilde{\theta}_\star}^{-1}
    \end{align*}
    with 
    \begin{multline*}
        V_{\tilde{\theta}_\star, \lambda}
    := \operatorname{Var} \left(\nabla \ell(X; {\tilde{\theta}_\star}) - \frac{1}{R}\sum_{r=1}^R \lambda_r \nabla \ell(f(O\circ m_r); {\tilde{\theta}_\star})\Bigg|M = \bm{1}_d\right)\\
     + \sum_{r=1}^R \left(\frac{\lambda_r}{R}\right)^2 \frac{p_0}{p_r} \Var\left(\nabla \ell(f(O); {\tilde{\theta}_\star})\bigg| M = m_r\right). 
    \end{multline*}
\end{theorem}
\begin{proof}
    Here, we reiterate the proof for the IPI estimator under MCAR, noting where the MCAR assumption can be replaced by the MCAR first moment assumption.

    Recall the decomposition
    \begin{align*}
        \hat{\theta}_\IPI - \tilde{\theta}_\star 
        &= \hat{\theta}_n - \tilde{\theta}_\star\\
        &\hspace{3mm} - [\nabla^2 L_\IPI(\hat{\theta}_n; \hat{\lambda})]^{-1}\left(\nabla L_\IPI (\hat{\theta}_n; \hat{\lambda}) - \nabla L_\IPI (\tilde{\theta}_\star; \hat{\lambda})\right)\\
        &\hspace{3mm} - [\nabla^2 L_\IPI(\hat{\theta}_n; \hat{\lambda})]^{-1}\left(\nabla L_\IPI (\tilde{\theta}_\star; \hat{\lambda}) - \nabla L_\IPI (\tilde{\theta}_\star;  {\lambda})\right)\\
        &\hspace{3mm} - [\nabla^2 L_\IPI(\hat{\theta}_n; \hat{\lambda})]^{-1}\nabla L_\IPI (\tilde{\theta}_\star; {\lambda}).
    \end{align*}

    By local Lipschitzness of the Hessian and law of large numbers, $\nabla^2 L_\IPI (\hat{\theta}_n; \hat{\lambda}) = \tilde{J}_{\tilde{\theta}_\star} + o_p(1)$.

    The stochastic Taylor expansion argument is the same as in the proof for Theorem \ref{thm:ipi_clt_mcar} in Section \ref{app:proof_ipi_clt_mcar} with $\nabla L_\IPI (\hat{\theta}_n; \hat{\lambda}) - \nabla L_\IPI (\tilde{\theta}_\star; \hat{\lambda}) = \tilde{J}_{\tilde{\theta}_\star}(\hat{\theta}_n - \tilde{\theta}_\star) + o_p(1/\sqrt{n})$.

By the MCAR first moment assumption,  for every $r \in [R]$, the third term accounting for randomness in the estimation of $\hat{\lambda}$ obeys  
    \[
        \mathbb{P}_{\tilde{N}_r} \left[\nabla \ell(f(O); \tilde{\theta}_\star)\right] - \mathbb{P}_{n} \left[\nabla \ell(f(O \circ m_r); \tilde{\theta}_\star)\right] = O_p(1/\sqrt{n}).
    \]
    Thus, by consistency of $\hat{\lambda}$ for $\lambda$, the entire third term is $o_p(1/\sqrt{n})$.

    The last term follows from the conditional CLT. The conditional mean is  
    \begin{multline*}
        \mathbb{E} \left[\nabla L_\IPI (\tilde{\theta}_\star; \lambda) |\mathcal{F}^M\right] 
        = \mathbb{E} \left[\nabla \ell(X; \tilde{\theta}_\star) | M = \bm{1}_d\right] + \frac{1}{R}\sum_{r =1}^R \lambda_r \bigg(\mathbb{E} \left[\nabla \ell(f(O); \tilde{\theta}_\star) | M = m_r) \right] \\
        - \mathbb{E} \left[\nabla \ell(f(O \circ m_r); \tilde{\theta}_\star) | M = \bm{1}_d\right]\bigg)= 0,
    \end{multline*}
    The second equality follows from the MCAR first moment (Assumption \ref{assump:MCAR_first_moment}) and $\tilde{\theta}_\star$ minimizing the smooth convex loss $\mathbb{E}[\nabla \ell(X; \tilde{\theta}_\star)| M = \bm{1}_d]$. A quick calculation leveraging conditional independence of the sample means by missing pattern gives the conditional variance:
    \begin{align*}
        n\Var \left(\nabla L_\IPI (\tilde{\theta}_\star; \lambda)  | \mathcal{F}^M\right) = V_{\tilde{\theta}_\star, \lambda} + o_p(1).
    \end{align*}

    Putting these findings all together gives
    \begin{align*}
        \sqrt{n} (\hat{\theta}_\IPI - \tilde{\theta}_\star) = o_p(1) - \tilde{J}_{\tilde{\theta}_\star}^{-1}(\sqrt{n} V_{\tilde{\theta}_\star, \lambda}^{1/2})V_{\tilde{\theta}_\star, \lambda}^{-1/2}\nabla L_\IPI(\tilde{\theta}_\star; \lambda) \overset{d}{\rightarrow} \mathcal{N}(0, \Sigma). 
    \end{align*}
\end{proof}

\subsubsection{Diagnostics}\label{app:diagnostics}
In this section, we quantify the asymptotic behavior of the test statistic we introduced, namely, 
\begin{align*}
    \hat{T}_\IPI = \frac{1}{R}\sum_{r =1}^R \hat{\lambda}_r \left(\mathbb{P}_{\tilde{N}_r} [\nabla \ell(f(O); \hat{\theta}_n)] - \mathbb{P}_n [\nabla \ell(f(O \circ m_r); \hat{\theta}_n)]\right),
\end{align*}
where $\hat{\lambda}$ is consistent for power-tuning parameter the $\lambda_\star$.

Proof of Proposition \ref{prop:asymp_behav_ipi_test}.
\begin{proof}
\begin{align*}
    \hat{T}_\IPI 
    &= \frac{1}{R}\sum_{r =1}^R \hat{\lambda}_r \left(\mathbb{P}_{\tilde{N}_r} \left[\nabla \ell(f(O); \hat{\theta}_n)\right] - \mathbb{P}_n \left[\nabla \ell(f(O \circ m_r); \hat{\theta}_n)\right]\right)\\
    &= \frac{1}{R}\sum_{r =1}^R \hat{\lambda}_r \left(\mathbb{P}_{\tilde{N}_r} \left[\nabla \ell(f(O); \tilde{\theta}_\star)\right] - \mathbb{P}_n \left[\nabla \ell(f(O \circ m_r); \tilde{\theta}_\star)\right]\right)\\
    &\hspace{3mm} + \frac{1}{R}\sum_{r =1}^R \hat{\lambda}_r \left(\mathbb{P}_{\tilde{N}_r} \left[\nabla \ell(f(O); \hat{\theta}_n) - \nabla \ell(f(O); \tilde{\theta}_\star)\right]\right)\\
    &\hspace{3mm} - \frac{1}{R}\sum_{r =1}^R \hat{\lambda}_r \left(\mathbb{P}_n \left[\nabla \ell(f(O \circ m_r); \hat{\theta}_n) - \nabla \ell(f(O \circ m_r); \tilde{\theta}_\star)\right]\right).
\end{align*}

By local Lipschitzness of the Hessian, the second term simplifies
\begin{align*}
    &\frac{1}{R}\sum_{r =1}^R \hat{\lambda}_r \left(\mathbb{P}_{\tilde{N}_r} \left[\nabla \ell(f(O); \hat{\theta}_n) - \nabla \ell(f(O); \tilde{\theta}_\star)\right]\right) \\
    &= \frac{1}{R}\sum_{r =1}^R \hat{\lambda}_r \left(\mathbb{P}_{\tilde{N}_r} \left[\nabla^2 \ell(f(O);\tilde{\theta}_\star)\right] + o_p(1)\right)(\hat{\theta}_n - \tilde{\theta}_\star) \\
    &= \left(\frac{1}{R}\sum_{r =1}^R \lambda_r \mathbb{E} \left[\nabla^2 \ell(f(O); \tilde{\theta}_\star) \bigg| M = m_r\right] + o_p(1)\right) (\hat{\theta}_n - \tilde{\theta}_\star).
\end{align*}

Similarly, the third term simplifies to
\[
    \left(\frac{1}{R}\sum_{r =1}^R \lambda_r \mathbb{E}\left[\nabla^2 \ell( f(O \circ m_r); \tilde{\theta}_\star) \bigg| M = \bm{1}_d\right] + o_p(1)\right)(\hat{\theta}_n - \tilde{\theta}_\star).
\]

By the smoothness and convexity of $\ell(X; \theta)$, $\hat{\theta}_n - \tilde{\theta}_\star = -H_{\tilde{\theta}_\star}^{-1} \mathbb{P}_n[\nabla \ell(X; \tilde{\theta}_\star)] + o_p(1/\sqrt{n}).$

Then, 
\begin{align*}
    \sqrt{n}\hat{T}_{\IPI} &= \frac{1}{R}\sum_{r =1}^R \lambda_r \mathbb{P}_{\tilde{N}_r} \left[\nabla \ell(f(O); \tilde{\theta}_\star\right] \\
    &\hspace{3mm}- \mathbb{P}_n \left[\frac{1}{R}\sum_{r =1}^R \lambda_r \left(\left(\tilde{H}_{\tilde{\theta}_\star}^r - {H}_{\tilde{\theta}_\star}^r\right)H_{\tilde{\theta}_\star}^{-1}\nabla \ell(X; \tilde{\theta}_\star)  + \nabla \ell(f(O \circ m_r); \tilde{\theta}_\star)\right)\right].
\end{align*}

The proposition follows from conditional CLT and explicit variance calculation.
\end{proof}

The asymptotic behavior for $\hat{T}_\IPI$ justifies a valid asymptotic p-value under the null of the MCAR first moment assumption.

\subsubsection{Diagnostics comparisons}\label{app:expt-diagnostics}
To test the diagnostic tool discussed in Section \ref{sec:moving_away_mcar} and analyzed in Section \ref{app:diagnostics}, we study the distribution of the resulting p-values in various settings for the factor model. In particular we keep $n = 200$, $N = 2000$, and $\Sigma = FF^T+ \Psi$ the same as in the original \textbf{factor model} setting with factors explaining $50\%$ of the variance. Here, we study the resulting p-values from $\hat{T}_\IPI$ and $\hat{T}_\mathrm{full}$ for different violations of the MCAR first moment assumption. Unless otherwise specified, there are the same $R = 10$ nontrivial missing patterns as in Experiment 1 and results are reported for 1000 trials.

\textbf{Additive shifts on imputed gradient depending on missing pattern.}
In this setting, we mimic different imputation behaviors based on the missingness pattern. In particular, we introduce a pattern-specific constant bias $c_r$:
\begin{align*}
    \nabla \ell (f(O \circ m_r); \hat{\theta}_n) = \begin{cases}
         \nabla \ell (f(O \circ m_r); \hat{\theta}_n) + c_r, &\text{ if $M = m_r$},\\
        \nabla \ell (f(O \circ m_r); \hat{\theta}_n), &\text{$M = \bm{1}_d$}.
    \end{cases}.
\end{align*}
While systematic biases are rarely going to behave like a constant shift, this setting allows a controlled environment to explore which settings are more or less detectable by $\hat{T}_\IPI$ and $\hat{T}_\text{full}.$ In particular, we consider three settings: (a) Null: $c_r = 0$ for all patterns $r$ so that the first-moment MCAR assumption is indeed satisfied; however, the total number of patterns $R$ increases, thus increasing the dimension of $\hat{T}_\text{full}$. (b) Uniform shift: $c_r = c$ for $c = 0.01, 0.05, 0.10$ observes the same additive shift across all missingness patterns. (c) Single-pattern shift: $c_{\tilde{r}} = 0$ for all $\tilde{r} \neq r$ and $c_r = 0.1$ for each of $r = 1, ..., 10$. That is, only one pattern exhibits a 0.1 additive shift. Results are presented in Figure \ref{fig:diagnostics-comparison}.

\begin{figure}[h!]
    \centering
    \includegraphics[width=\linewidth]{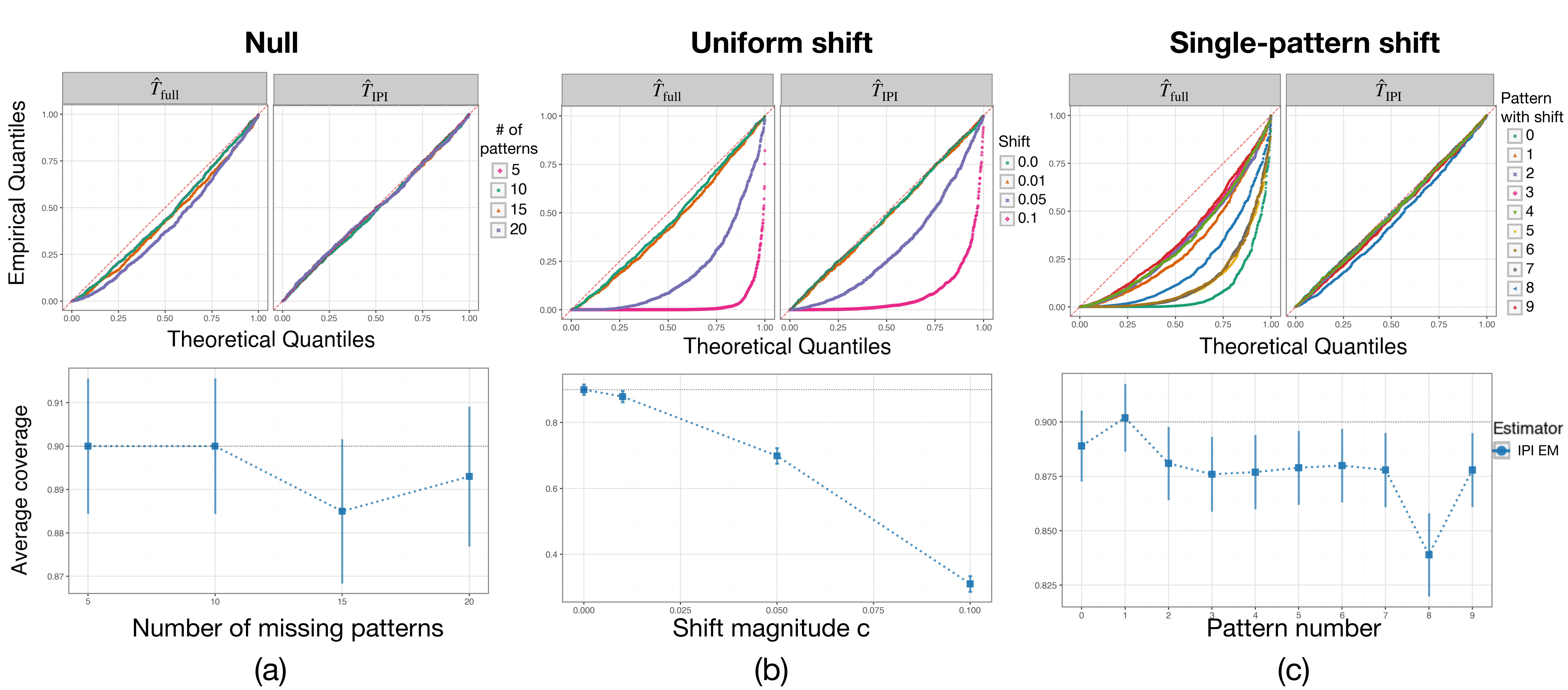}
    \caption{(Upper row) QQ plot of p-values for 100 runs corresponding to each test. (Lower row) Average coverage rates of the IPI method corresponding to each run. (a) Results when the first-moment MCAR assumption holds. (b) Results for uniform additive shifts across all missingness patterns, adjusting the magnitude of the shift. (c) Results for 0.1 additive shift only added to a single pattern. }
    \label{fig:diagnostics-comparison}
\end{figure}

In Figure \ref{fig:diagnostics-comparison}(a), we observe that even when the first-moment MCAR assumption holds, the p-values from the higher dimensional $\hat{T}_\mathrm{full}$ is sub-uniform while those from $\hat{T}_\IPI$ remain uniform. Under uniform shift (Figure \ref{fig:diagnostics-comparison}(b)), both diagnostics are able to detect large deviations from the MCAR first moment assumption. Lastly, in Figure \ref{fig:diagnostics-comparison}(c), we observe that $\hat{T}_\mathrm{full}$ is highly sensitive to sparse shifts for each pattern. In comparison, $\hat{T}_\IPI$ really only is sub-uniform with pattern 8, which in previous experiments was shown to be the best performing pattern and to have a large power-tuning $\lambda$ weight. From this exploration, $\hat{T}_\IPI$ p-values could be more aligned with a practitioner's goal. More rigorous evaluation of the finite-sample type I error and power of these tests is left for future work. 

\subsection{General CIPI}
We now evaluate CIPI under general missingness.

\begin{definition}[CIPI estimator, general]
    $\hat{\theta}_{\CIPI, \hat{\lambda}} = \hat{\theta}_n - \hat{H}_n^{-1} \nabla L_\CIPI (\hat{\theta}_n; \hat{\lambda})$, where $\hat{H}_n {=} \tilde{J}_{\tilde{\theta}_\star, \bar{f}} + o_p(1)$ with
    \begin{multline*} 
    \tilde{J}_{\tilde{\theta}_\star, \bar{f}} := \mathbb{E}\left[\nabla^2 \ell(X; \tilde{\theta}_\star) -\frac{1}{R}\sum_{r = 1}^R \lambda_r \nabla^2 \ell(\bar{f}(O \circ m_r); \tilde{\theta}_\star) \bigg| M = \bm{1}_d \right] \\ +  \frac{1}{R}\sum_{r =1}^R \lambda_r \mathbb{E}\left[\nabla^2 \ell(\bar{f}(O); \tilde{\theta}_\star) \bigg| M = m_r \right],\end{multline*}
    where we assume $\tilde{J}_{\tilde{\theta}_\star, \bar{f}} \succ 0.$ In the following exposition, we will take $ \hat{H}_n = \nabla^2 L_\CIPI(\hat{\theta}_n; \hat{\lambda})$.
\end{definition}

The main assumption we make for CIPI under general missingness is the following.
\begin{assumption}[MCAR first moment for CIPI]\label{assump:mcar-firstmoment-cipi}
For every missing pattern \(m_r\;(r\in[R])\) the mean imputed score is
unchanged by masking, both fold-wise and in the long-run average model:
\begin{align*}
    \frac{1}{K}\sum_{k =1}^K \mathbb{E} \left[\nabla \ell(f^{(k)}(O); \theta_\star) \bigg| M = m_r, f^{(k)}\right] &= \frac{1}{K}\sum_{k =1}^K \mathbb{E} \left[\nabla \ell(f^{(k)}(O \circ m_r); \theta_\star) \bigg| M = 
    \bm{1}_d,f^{(k)} \right],\\
    \mathbb{E}\left[\nabla \ell (\bar{f}(O); \theta_\star) | M = m_r\right] &= \mathbb{E}\left[\nabla \ell(\bar{f}(O \circ m_r); \theta_\star) | M = \bm{1}_d\right].
\end{align*}
\end{assumption}

Since removing the MCAR assumption makes the Hessian term harder to interpret, we introduce an additional stability assumption. This assumption restores a similar interpretability to that found in IPI under MCAR first-moment missingness.
\begin{assumption}[Stability of the Hessian]\label{assump:cipi-hessian-stability}
For us to better interpret the cross-fitted Hessian, we add the extra stability assumption that for each $r \in [R]$ and $\tilde{r} \in \{0, r\}$,
\begin{align*}
    \mathbb{E}[\nabla^2 \ell(f^{(1)}(O \circ m_r); \theta) - \nabla^2 \ell(\bar{f}(O \circ m_r); \theta)|f^{(1)}, M= m_{\tilde{r}}] = o_p(1),
\end{align*}
where again, $\bar{f}$ can be thought of as the long-run average of the fitted imputor on $N(K-1)/K$ samples.
\end{assumption}

\begin{theorem}\label{thm:general-cipi}
Assume regularity (Assumption \ref{assump:cipi_regularity}), stability (Assumptions \ref{assump:stability assumptions} and \ref{assump:cipi-hessian-stability}), conditional Lindeberg condition (Assumption \ref{assump:cond_lindeberg_cipi}), and Assumption \ref{assump:mcar-firstmoment-cipi}. Then, for $\hat{\lambda} = \lambda + o_p(1)$,
\[
    \sqrt{n}\left(\hat{\theta}_{\CIPI, \hat{\lambda}} - \tilde{\theta}_\star\right) \overset{d}{\rightarrow} \mathcal{N}(0, \bar{\Sigma}), \quad 
    \bar{\Sigma} = \tilde{J}_{\tilde{\theta}_\star,\bar{f}}^{-1} V_{\tilde{\theta}, \lambda} \tilde{J}_{\tilde{\theta}_\star,\bar{f}}^{-1}
\]
and $V_{\tilde{\theta}_\star, \lambda}$ is as defined in Eq.~\eqref{eq:cipi_asymp_var}.
\end{theorem}
\begin{proof}
    As in the proof for Theorem \ref{thm:general_ipi}, we functionally reiterate the proof, highlighting where we replace MCAR with MCAR first moment assumptions.

    Recall the decomposition in Equations \eqref{eq:cipi-decomp1} and \eqref{eq:cipi-decomp2}. 
    The stochastic Taylor expansion of $\nabla L_\CIPI$ around $\tilde{\theta}_\star$ argument stays the same; the second term capturing the uncertainty in $\hat{\lambda}$ is $o_p(1/\sqrt{n})$ by the fold-wise MCAR first moment assumption; and the last term continues to be $o_p(1/\sqrt{n})$ from stability (Assumption \ref{assump:stability assumptions}) and the MCAR first moment assumption for both the fold-wise and long-run average imputations.

    Therefore, by conditional central limit theorem, 
    \begin{align*}
        \hat{\theta}_{\CIPI, \hat{\lambda}} - \tilde{\theta}_\star = -\left[\nabla^2 L_{\CIPI} (\hat{\theta}_n; \hat{\lambda})\right]^{-1}\nabla L_\IPI^{\bar{f}}(\theta_\star; \lambda) + o_p(1/\sqrt{n}) \overset{d}{\rightarrow} \mathcal{N}\left(0,\bar{\Sigma}\right).
    \end{align*}
\end{proof}

\section{Experimental details}\label{app:exptl_details}
In this section, we provide further experimental details supplementing the discussion in the main text. 

\subsection{Baseline method details}
\paragraph{Complete-case method.}
The complete-case method is based on the asymptotics of the minimizer on the fully observed samples: 
\[
    \hat{\theta}_n = \argmin_\theta \mathbb{P}_n[\nabla \ell(X; \theta)].
\]
\paragraph{AIPW baseline.}
The AIPW baseline refers to the approach discussed in Appendix \ref{app:aipw_expts} (see Eq \eqref{app-eq:aipw_class1_augmentation} for details). This adapts the Class 1 estimators in \citet{tsiatis2006semiparametric} to the setting where the restricted moment condition $E[Y|X] = X^\top \theta$ may not hold and instead we target the estimand $\theta_\star$ satisfying $E[X(Y- X^\top \theta_\star] = 0$.

\paragraph{Naive single imputation.}
We first create a single complete dataset \(\tilde X\) by imputing missing entries, then perform inference as if \(\tilde X\) were fully observed (ignoring imputation uncertainty). For an M-estimator with loss \(\ell\),
\[
    \hat X = \begin{cases}
        X,  & \text{ if } M = \mathbf{1}_d,\\
        f(O), & \text{ otherwise };
    \end{cases} \qquad
    \hat{\theta}_N = \arg\min_{\theta}\, \mathbb{P}_N\!\left[\ell(\tilde X;\theta)\right].
\] 

\begin{itemize}
  \item \textbf{EM (factor model).} Fit a factor model to the incomplete data via EM to estimate the covariance; impute missing entries by the model’s conditional means.
  \item \textbf{MissForest (LightGBM).} Iteratively impute each variable using tree-based models; here \texttt{LightGBM} replaces random forests for speed~\cite{stekhoven2012missforest}.
  \item \textbf{Hot-deck.} For each missing entry, copy the value from a “donor” record—typically the nearest neighbor with that value observed~\cite{andridge2010review_hotdeck}.
  \item \textbf{Mean imputation.} Replace each missing entry with the column mean computed from observed values.
  \item \textbf{Zero imputation.} Replace missing entries with \(0\).
\end{itemize}

\paragraph{Single-pattern IPI.} The single-pattern IPI baseline method is the IPI method applied to a single missing pattern. That is, single-pattern IPI for missingness pattern $m_r$ applies to the loss:
\[
    L_\IPI^r(\theta) = \mathbb{P}_n \left[\nabla \ell(X; {\theta})\right] + \tilde{\lambda}_r \left(\mathbb{P}_{\tilde{N}_r} \left[\nabla \ell(f(O); {\theta})\right] - \mathbb{P}_{n} \left[\nabla \ell(f(O \circ m_r); {\theta})\right]\right),
\]
where $\tilde{\lambda}_r$ may be different from the $\lambda_r$ for full-data IPI. We consider this a baseline comparison because under semi-supervised settings for the linear regression setting, this single-pattern IPI one-step estimator aligns exactly with the PPI and PPI++ estimator.

\subsection{Factor model details}\label{app:fm_expts_more}
In this section, we expand upon the main paper results to provide a more thorough exploration of the factor model setting.

\subsubsection{Imputation methods and naive baselines}
In the same setting as Experiment 1 in Section \ref{expt:mcar}, we compare the effective sample size and average coverage over 1000 trials of IPI with naive single imputation approaches, using the imputation methods listed above -- namely, EM, MissForest, hot-deck, mean, and zero. We also include the AIPW baseline in our comparisons. Note that the color-coding is different from most of the figures as colors are now used to distinguish between different imputation methods (e.g. EM vs. MissForest). As seen in Figure  \ref{fig:fm_ipi_classical_full_comparison}, the IPI methods achieve roughly nominal coverage while their single imputation counterparts fail to achieve nominal coverage. This illustrates the usefulness of IPI for capturing the bias of single imputation methods. Moreover, the efficiency of IPI depends on the imputation method; however, we see that they all observe improved efficiency over the complete-case approach.
\begin{figure}[H]
    \centering
    \includegraphics[width=\linewidth]{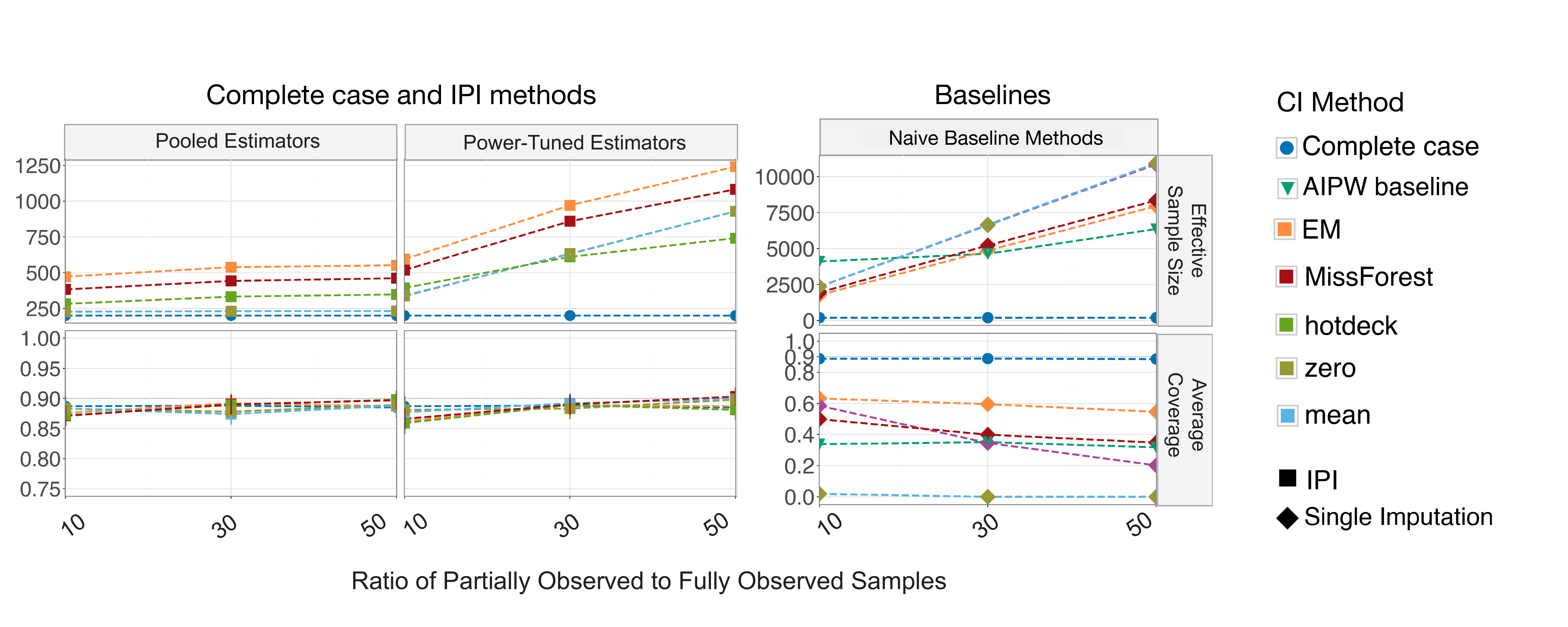}
    \caption{Full comparison of IPI methods with their naive single imputation baselines across five different imputation methods. AIPW and complete-case methods are provided also for completion.}
    \label{fig:fm_ipi_classical_full_comparison}
\end{figure}

\subsubsection{Non-MCAR experiment}
To illustrate the importance of being careful about which estimand one is targeting when we move away from non-MCAR settings, we provide an illustrative experiment where $\tilde{\theta}_\star$ progressively shifts further away from $\theta_\star$. In particular, we simulate data in the factor model setting with $R = 3$ missingness patterns, in which $\tilde{\theta}_\star = c \cdot \theta_\star$ where $c = \{1.05, 1.10, 1.15, \hdots, 1.50\}$. We report sample efficiency, average coverage of $\theta_\star$, and average coverage of $\tilde{\theta}_\star$ over 1000 trials.
\begin{figure}[h!]
    \centering
    \includegraphics[width=0.7\linewidth]{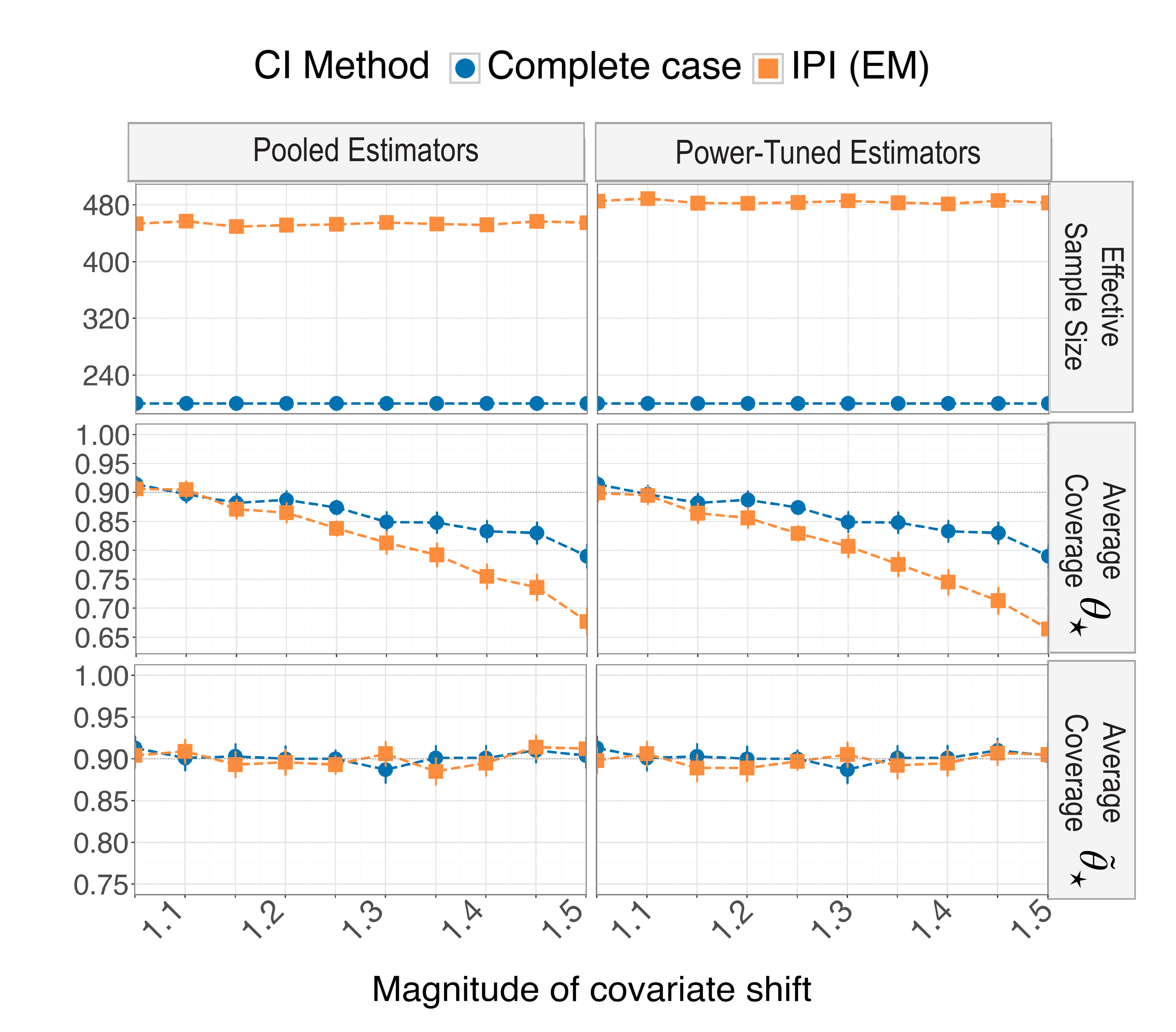}
    \caption{Effective sample size (top row); average coverage for $\theta_\star$ (middle row); and average coverage for $\tilde{\theta}_\star$ (bottom row), when $\tilde{\theta}_\star \neq \theta_\star$ trials for complete-case and IPI EM methods for the factor model.}
    \label{fig:nonmcar}
\end{figure}

In Figure \ref{fig:nonmcar}, we observe that as the estimand's shift magnitude becomes stronger, IPI methods suffer from undercoverage of $\theta_\star$ with IPI EM being more sensitive to undercoverage in these settings; however, note that both IPI and complete-case methods roughly achieve nominal coverage for $\tilde{\theta}_\star$.

\subsection{Census survey}
The original ACS Census survey \cite{census_survey_folktables} is comprised of 380,091 samples and 288 features. To align with previous PPI literature, we limited ourselves to the set of features used in previous analyses. After filtering out any partially observed datapoints, so we still have access to the ground truth on the finite population of fully observed datapoints, we get 255,111 samples and 17 features (Table \ref{tab:acs-feature-codes}).  We use a random subsample of these data for each experimental trial (using random seeds 0 through 999).

\begin{table}[H]
  \caption{Census feature codes. More information provided at the \href{https://usa.ipums.org/usa/resources/codebooks/DataDict1822.pdf}{ACS data dictionary}.}
  \label{tab:acs-feature-codes}
  \centering
  \begin{tabular}{lll}
    % \toprule
    % \multicolumn{2}{c}{Part}                   \\
    \cmidrule(r){1-3}
    \textbf{Name}     & \textbf{Description}    & \textbf{Data type}  \\
    \midrule
    AGEP     & age                          & scalar (0-99)\\
    SCHL     & educational attainment       & ordinal (0-24 by school year/degree)  \\
    MAR      & marital status               & categorical (5 values)   \\
    DIS      & disability                   & binary\\
    ESP      & employment status of parents & categorical\\
    CIT      & citizenship status           & categorical \\
    MIG      & mobility status              & categorical \\
    MIL      & military service             & categorical\\
    ANC1P    & ancestry                     & categorical (1000 values)\\
    NATIVITY & nativity                     & binary\\
    DEAR     & hearing difficulty           & binary\\
    DEYE     & vision difficulty            & binary\\
    DREM     & cognitive difficulty         & binary\\
    SEX      & sex                          & binary\\
    RAC1P    & race                         & categorical (9 values)\\
    COW      & class of worker              & categorical\\
    PINCP    & person's income              & scalar\\
    \bottomrule
  \end{tabular}
\end{table}

\textbf{Collecting missing patterns.} 
Due to missing values or masked values for privacy concerns, the ACS census survey contains allocation flags for each variable indicating whether the sample value was imputed. We thus tracked these flags in the partially observed dataset (124,980 samples with 8\% of the table entries imputed) and chose the top ten missing patterns with at least one missing feature for age, schooling level, or income. The resulting missing patterns are depicted in Figure \ref{fig:acs-results}(a).

\textbf{Training and experimental details.} 
For the imputation model, we adapted a Python implementation of MissForest (initial code available at \texttt{\href{https://github.com/yuenshingyan/MissForest}{https://github.com/yuenshingyan/MissForest}}) to better integrate with our codebase. We extended the original functionality to allow more flexible configuration of the gradient-boosted tree training parameters. Specifically, we used a learning rate of 0.02, 32 leaves, a bagging fraction of 90\%, bagging frequency of 100\%, feature fraction of 90\%, and a maximum tree depth of 6.

We did not report results from a cross-fitting experiment due to the substantial computational cost of re-fitting MissForest for the bootstrap variance estimation. This high cost stems from the use of an iterative imputation approach, where a LightGBM model is trained for each of 9 partially observed covariates—each trial involving at least 50 training iterations of 9 gradient-boosted trees.

\subsection{Allergen dataset}
For the allergen chip challenge dataset experiment, we focused on the data for one IgE tracking chip, which comprised of demographic, clinical, and IgE chip readings for 2,351 samples with 184 features (112 IgE chip readings).

\textbf{Pre-processing.}
We began by filtering the dataset to include only patients with detectable IgE levels, yielding 2,114 samples. We then selected clinical and demographic features: multiple allergies, conjunctivitis, asthma, month of assessment, average airway-related IgE, region in France, age, sex, and rhinitis severity. For IgE-specific analysis, we focused on peanut allergens Ara h 1, 2, 3, 6, and 8. Features are provided in Table \ref{fig:allergen-feature-codes}. Finally, to ensure sufficient data per missing pattern, we excluded patterns with fewer than 50 samples, resulting in a final dataset of 1,760 samples.

\begin{table}[H]\label{table:acs-feature-codes}
  \caption{Allergen features. More information provided at the \href{https://usa.ipums.org/usa/resources/codebooks/DataDict1822.pdf}{ACS data dictionary}.}
  \label{fig:allergen-feature-codes}
  \centering
  \begin{tabular}{lll}
    \cmidrule(r){1-3}
    \textbf{Name}     & \textbf{Description}    & \textbf{Data type}  \\
    \midrule
    AGE & age & scalar\\
    SEX & sex & binary\\
    REGION      & region in France & categorical\\
    MONTH      & month of collection                   & categorical \\
    POLY     & presence of multiple allergies                          & binary\\
    RHIN     & rhinitis severity               & ordinal (6 values)   \\
    PINK   & conjunctivitis/pink eye       & binary (0-24 by school year/degree)  \\
    ASTHM     & asthma               & ordinal (6 values)   \\
    COFACTOR      & co-factors          & categorical (9 values), each binary encoded  \\
    ANAPH & anaphylaxis category & categorical (9 values), each binary encoded\\
    AIRWAY\_IgE      & average airway/. IgE reading             & scalar\\
    Ara h 1,..., 8 IgEs   &   IgE reading         & scalar\\
    \bottomrule
  \end{tabular}
\end{table}

\begin{figure}[H]
    \centering
    \includegraphics[width=\linewidth]{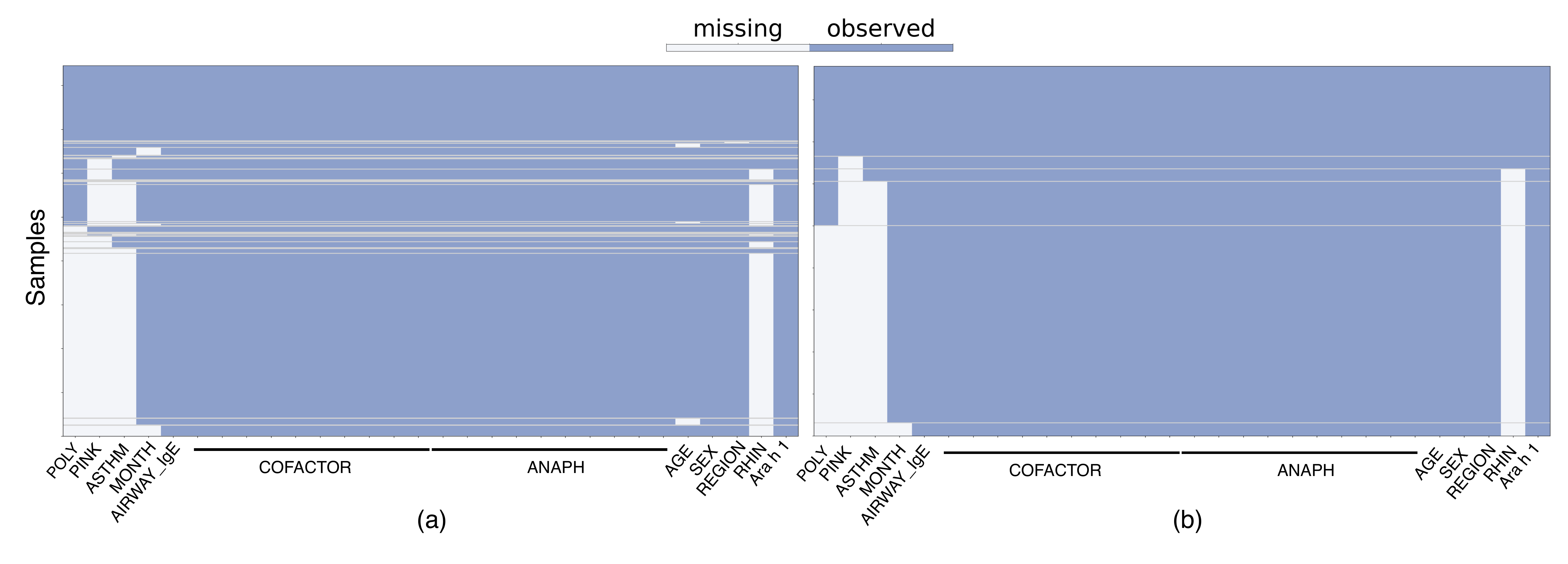}
    \caption{Missingness pattern ratios in the observed dataset after selecting this subset of features. (a) Prior to filtering out patterns with fewer than 50 samples. (b) After filtering out patterns with fewer than 50 samples.}
    \label{fig:allergendata_missing_patterns}
\end{figure}